\documentclass[12pt]{amsart}
\usepackage{amssymb}
\usepackage{color}
\usepackage{amsmath,epic,curves,amscd}
\usepackage[english]{babel}
\usepackage{graphicx}
\usepackage{comment}
\usepackage{appendix}
\usepackage{mathdots}
\usepackage[all]{xy}
\newtheorem{claim}{}[section]
\newtheorem{theorem}[claim]{Theorem}
\newtheorem{definition}[claim]{Definition}
\newtheorem{conjecture}[claim]{Conjecture}
\newtheorem{lemma}[claim]{Lemma}
\newtheorem{proposition}[claim]{Proposition}
\newtheorem{corollary}[claim]{Corollary}
\newtheorem{example}[claim]{Example}
\theoremstyle{remark}
\newtheorem{remark}[claim]{Remark}
\renewenvironment{proof}{\noindent{\it Proof. \hskip0pt}}
                      {$\square$\par\medskip}

\newcommand\re{{\text{\rm Re}}\,}   
\newcommand\xx{{\text{\sf X}}}
\newcommand\sym{{\text{\rm SY}}}
\newcommand\syms{{\text{\rm sy}}}
\newcommand\uni{{\text{\rm UN}}}
\newcommand\unis{{\text{\rm un}}}

\def\a{\alpha}
\def\b{\beta}
\def\g{\gamma}

\def\t{\theta}

\def\l{\lambda}

\def\p{\pi}

\def\s{\sigma}

\def\ps{\psi}

\def\diag{\mathop{\rm diag}}

\def\sx{{\sf X-}}

\def\ci{{\rm i}}

\def\ox{\otimes}

\newcommand{\nc}{\newcommand}

\nc{\cA}{{\cal A}} \nc{\cB}{{\cal B}} \nc{\cC}{{\cal C}}
\nc{\cD}{{\cal D}} \nc{\cE}{{\cal E}} \nc{\cF}{{\cal F}}
\nc{\cG}{{\cal G}} \nc{\cH}{{\cal H}} \nc{\cI}{{\cal I}}
\nc{\cJ}{{\cal J}} \nc{\cK}{{\cal K}} \nc{\cL}{{\cal L}}
\nc{\cM}{{\cal M}} \nc{\cN}{{\cal N}} \nc{\cO}{{\cal O}}
\nc{\cP}{{\cal P}} \nc{\cQ}{{\cal Q}} \nc{\cR}{{\cal R}}
\nc{\cS}{{\cal S}} \nc{\cT}{{\cal T}} \nc{\cU}{{\cal U}}
\nc{\cV}{{\cal V}} \nc{\cW}{{\cal W}} \nc{\cX}{{\cal X}}
\nc{\cZ}{{\cal Z}}

 \nc{\bbA}{\mathbb{A}} \nc{\bbB}{\mathbb{B}} \nc{\bbC}{\mathbb{C}}
 \nc{\bbD}{\mathbb{D}} \nc{\bbE}{\mathbb{E}} \nc{\bbF}{\mathbb{F}}
 \nc{\bbG}{\mathbb{G}} \nc{\bbH}{\mathbb{H}} \nc{\bbI}{\mathbb{I}}
 \nc{\bbJ}{\mathbb{J}} \nc{\bbK}{\mathbb{K}} \nc{\bbL}{\mathbb{L}}
 \nc{\bbM}{\mathbb{M}} \nc{\bbN}{\mathbb{N}} \nc{\bbO}{\mathbb{O}}
 \nc{\bbP}{\mathbb{P}} \nc{\bbQ}{\mathbb{Q}} \nc{\bbR}{\mathbb{R}}
 \nc{\bbS}{\mathbb{S}} \nc{\bbT}{\mathbb{T}} \nc{\bbU}{\mathbb{U}}
 \nc{\bbV}{\mathbb{V}} \nc{\bbW}{\mathbb{W}} \nc{\bbX}{\mathbb{X}}
 \nc{\bbZ}{\mathbb{Z}}

\def\bcj{\begin{conjecture}}
\def\ecj{\end{conjecture}}
\def\bcr{\begin{corollary}}
\def\ecr{\end{corollary}}
\def\bd{\begin{definition}}
\def\ed{\end{definition}}
\def\bea{\begin{eqnarray}}
\def\eea{\end{eqnarray}}
\def\bma{\begin{bmatrix}}
\def\ema{\end{bmatrix}}
\def\bem{\begin{enumerate}}
\def\eem{\end{enumerate}}
\def\bex{\begin{example}}
\def\eex{\end{example}}
\def\bim{\begin{itemize}}
\def\eim{\end{itemize}}
\def\bl{\begin{lemma}}
\def\el{\end{lemma}}
\def\bpf{\begin{proof}}
\def\epf{\end{proof}}
\def\bpp{\begin{proposition}}
\def\epp{\end{proposition}}
\def\bqu{\begin{question}}
\def\equ{\end{question}}
\def\br{\begin{remark}}
\def\er{\end{remark}}
\def\bt{\begin{theorem}}
\def\et{\end{theorem}}

\newcommand{\ket}[1]{|#1\rangle}
\newcommand{\proj}[1]{| #1\rangle\!\langle #1 |}

\newcommand{\abs}[1]{|#1|}

\newcommand{\vr}{\varrho}

\newcommand\lan{\langle}
\newcommand\ran{\rangle}
\newcommand\tr{{\text{\rm Tr}}\,}
\newcommand\ot{\otimes}
\newcommand\ttt{{\text{\rm t}}}

\begin{document}
\baselineskip 6.0 truemm
\parindent 1.5 true pc

\title{Separability criterion for three-qubit states\\ with a four dimensional norm}

\begin{abstract}
We give a separability criterion for three qubit states in terms of
diagonal and anti-diagonal entries. This gives us a complete
characterization of separability when all the entries are zero
except for diagonal and anti-diagonals. The criterion is expressed
in terms of a norm arising from anti-diagonal entries. We compute
this norm in several cases, so that we get criteria with which we
can decide the separability by routine computations.
\end{abstract}

\author{Lin Chen, Kyung Hoon Han and Seung-Hyeok Kye}
\address{Lin Chen, School of Mathematics and Systems Science, Beihang University, Beijing 100191, China}
\address{International Research Institute for Multidisciplinary Science, Beihang University, Beijing 100191, China}
\email{linchen at buaa.edu.cn}

\address{Kyung Hoon Han, Department of Mathematics, The University of Suwon, Gyeonggi-do 445-743, Korea}
\email{kyunghoon.han at gmail.com}
\address{Seung-Hyeok Kye, Department of Mathematics and Institute of Mathematics, Seoul National University, Seoul 151-742, Korea}
\email{kye at snu.ac.kr}
\thanks{LC was supported by Beijing Natural Science Foundation (4173076),
the NNSF of China (Grant No. 11501024), and the Fundamental Research Funds for the Central Universities
(Grant Nos. KG12001101, ZG216S1760 and ZG226S17J6).
Both KHH and SHK were partially supported by NRF-2017R1A2B4006655, Korea}

\subjclass{81P15, 15A30}

\keywords{three-qubit X-states, separability criterion, norm}

\maketitle


\section{introduction}

In the current quantum information and quantum computation theory,
the notion of entanglement is a very important resource. But, it is
very difficult to distinguish entanglement from separability, and a
complete characterization is known in the literature for very few
cases. For example, separability for $2\otimes 2$, $2\otimes 3$
states and $2\otimes n$ states with low ranks is known to be
equivalent to positivity of partial transposes
\cite{{stormer},{woronowicz},{choi-ppt},{peres},{horo-1},{2xn}}. In
the three qubit cases, separability of Greenberger-Horne-Zeilinger
diagonal states has been completely characterized recently by the
second and third authors \cite{han_kye_GHZ}. Note that separability
problem is known to be an $NP$-hard problem in general
\cite{gurvits}. We recall that a state is said to be separable if it
is the convex sum of pure product states. A state which is not
separable is called entangled.

GHZ diagonal states are typical kinds of {\sf X}-states, the states whose entries are zero except for
diagonal and anti-diagonal entries. Multi-qubit {\sf X}-states arise in various contexts in quantum
information theory. See \cite{{dur},{abls},{yu},{rau},{vin10},{wein10},{mendo}}, for example. The purpose of this paper
is to provide a complete characterization for the separability of three qubit {\sf X}-states. Because the {\sf X}-part of a
separable three qubit state is again separable \cite{han_kye_GHZ}, our characterization gives rise to
a necessary separability criterion in terms of diagonal and anti-diagonal entries for general three qubit states.

G\" uhne \cite{guhne_pla_2011} gave a necessary condition for separability of three qubit states,
together with a numerical evidence that this is also sufficient for GHZ diagonal states.
On the other hand, the second and third authors gave a characterization of
three qubit {\sf X}-shaped entanglement witnesses \cite{han_kye_tri}.
They also proved \cite{han_kye_GHZ} that a complicated expression in the G\" uhne's criterion
can be replaced by a simpler formula in \cite{han_kye_tri}, and
the G\" uhne's criterion is actually sufficient for separability of GHZ diagonal states.

In this paper, we first prove that G\" uhnes's criterion gives rise to a sufficient condition for separability
of general three qubit {\sf X}-states. In this characterization, the quantity determined by anti-diagonal entries
turns out to be the dual norm of a norm for $\mathbb C^4$ which is determined by the phase difference
in \cite{han_kye_phase} as well as magnitudes of entries.
See below for the definition of phase difference. In this way, we see that the phase difference
plays a key role, as it was found in \cite{han_kye_phase}.

We consider three qubit objects as $8\times 8$ matrices with the lexicographic order of indices,
and so a three qubit self-adjoint {\sf X}-shaped
matrix is of the form
$$
X(a,b,c)= \left(
\begin{matrix}
a_1 &&&&&&& c_1\\
& a_2 &&&&& c_2 & \\
&& a_3 &&& c_3 &&\\
&&& a_4&c_4 &&&\\
&&& \bar c_4& b_4&&&\\
&& \bar c_3 &&& b_3 &&\\
& \bar c_2 &&&&& b_2 &\\
\bar c_1 &&&&&&& b_1
\end{matrix}
\right),
$$
for vectors
$a=(a_1,a_2,a_3,a_4),b=(b_1,b_2,b_3,b_4)\in\mathbb R^4$ and
$c=(c_1,c_2,c_3,c_4)\in\mathbb C^4$.
We recall that a state is GHZ diagonal if and only if it is an {\sf X}-state of the form $X(a,b,c)$
with $a=b$ and $c\in\mathbb R^4$.
We also write $c_i=r_ie^{{\rm
i}\theta_i}$ with $r_i\ge0$ and $\theta_i\in\mathbb R$. The
separability of $\varrho=X(a,b,c)$ depends heavily on the {\sl phase
difference} $\phi_\varrho$ defined by
$$
\phi_\varrho=(\theta_1+\theta_4)-(\theta_2+\theta_3),\mod 2\pi
$$
as it was discussed in \cite{han_kye_phase}. This is determined by the anti-diagonal
part $c\in \mathbb C^4$, and will be denoted by $\phi_c$ sometimes.

G\" uhne \cite{guhne_pla_2011} showed that if a three qubit state $\varrho$ with its {\sf X}-part $X(a,b,c)$ is separable then the inequality
\begin{equation}\label{kye-main-ineq-sep}
C(z)\Delta_\varrho\ge \mathcal L (\varrho, z)
\end{equation}
holds for every $z\in \mathbb C^4$, where
$$
\begin{aligned}
\Delta_\varrho:=&\min\{\sqrt{a_ib_i}\,\,(i=1,2,3,4),\ \sqrt[4]{a_1b_2b_3a_4},\sqrt[4]{b_1a_2a_3b_4}\},\\
\mathcal L (\varrho, z) := & {\rm Re} \left(z_1 c_1 + z_2 c_2 + z_3 c_3 + z_4 \bar c_4 \right), \\
C(z):= & \sup_{\alpha,\beta,\gamma} |{\rm Re} (z_1) \cos (\alpha+\beta+\gamma) - {\rm Im} (z_1) \sin (\alpha+\beta+\gamma)
+ {\rm Re} (z_2) \cos (\alpha) - {\rm Im} (z_2) \sin (\alpha) \\
&  + {\rm Re} (z_3) \cos (\beta) - {\rm Im} (z_3) \sin (\beta)
+ {\rm Re} (z_4) \cos (\gamma) - {\rm Im} (z_4) \sin (\gamma)|.
\end{aligned}
$$
On the other hand, it was shown in \cite{han_kye_tri} that
a three qubit {\sf X}-shaped non-positive self-adjoint matrix $W=X(s,t,u)$ is an entanglement witness
if and only if the equality $A(s,t)\ge B(u)$ holds, where
\begin{equation}\label{a_st}
\begin{aligned}
A(s,t)
&:=\inf_{r>0} \left[\sqrt{(s_1 r^{-1}+t_4r)(s_4r^{-1}+t_1r)} +\sqrt{(s_2r^{-1}+t_3r)(s_3r^{-1}+t_2r)}\right],\\
B(u) &:=\max_\theta\left(|u_1e^{{\rm i}\theta} +\bar u_4|+|u_2
e^{{\rm i}\theta} + \bar u_3|\right).
\end{aligned}
\end{equation}
Later, it was shown in \cite{han_kye_GHZ} that the complicated formula of $C(z)$ is simplified as
\begin{equation}\label{CB}
C(z)=B(z_1,z_2,z_3,\bar z_4)
\end{equation}
and G\" uhne's condition (\ref{kye-main-ineq-sep}) is actually sufficient for the separability
of three qubit GHZ diagonal states. In the next section, we show that the inequality
(\ref{kye-main-ineq-sep}) actually characterizes separability of general three qubit {\sf X}-shaped states.

It is important to note that $u\mapsto B(u)$ defines a norm on the four dimensional space $\mathbb C^4$,
and so we will use the norm notation
\begin{equation}\label{kye_norm-def}
\| z\|_\xx=\max_\sigma \left(|z_1 e^{{\rm i}\sigma}+\bar z_4| + |z_2
e^{{\rm i}\sigma} +\bar z_3|\right),\qquad x\in\mathbb C^4,
\end{equation}
for the number $B(z)$. The norm $\|z\|_\xx$ is determined by the phase difference $\phi_z$ and the magnitudes $|z_i|$.
We will see in Section \ref{kye-section-norm} that a state $\varrho=X(a,b,c)$ is separable
if and only if the inequality
\begin{equation}\label{main-cri-dual-norm,,,}
\Delta_\varrho\ge \|c\|_\xx^\prime
\end{equation}
holds with the dual norm $\|\ \|_\xx^\prime$ of the norm $\|\ \|_\xx$ in (\ref{kye_norm-def}).
We exhibit some elementary properties of the norm $\|\ \|_\xx$ and its dual
norm $\|\ \|_\xx^\prime$. They have some interesting properties:
They are determined by the phase differences and magnitudes of entries;
they are invariant under eight kinds of permutations; the norm $\|\
\|_\xx$ is decreasing as the phase difference increase from $0$ to $\pi$.

In Section \ref{sec:norm}, we compute the dual norm $\|c\|_\xx^\prime$ in several cases:
\begin{itemize}
\item
all the entries are real,
\item
at least one of the entries is zero,
\item
entries are partitioned into two groups of two entries with common magnitudes.
\end{itemize}
The dual norm has an obvious lower bound $\|c\|_\xx^\prime\ge \|c\|_\infty$. We also provide a condition
under which we have $\|c\|_\xx^\prime= \|c\|_\infty$. These results give us complete
operational separability criteria for those {\sf X}-states.
In Section \ref{sec:norm-estimate}, we give several estimates for the norms
$\|\ \|_\xx$ and $\|\ \|_\xx^\prime$. Among them, lower bounds for $\|\ \|_\xx^\prime$
will give rise to necessary criteria for the separability of general three qubit states.

One merit of our separability criterion is that we do not have to decompose a state into the sum of pure product states
in order to show that a given state is separable. This was possible through the duality
between tri-partite separability and positivity of bi-linear maps \cite{kye_3qb_EW}.
Nevertheless, it is another interesting problem to get a decomposition, even though we
already know that a given state is separable. We close the paper to discuss this problem.


\section{Separability criterion of \sx states}
\label{guhne-condi-}

This section will be devoted to prove the following:

\begin{theorem}\label{diagonal}
A three qubit {\sf X}-state $\varrho=X(a,b,c)$ is separable if and only if
the inequality {\rm (\ref{kye-main-ineq-sep})} holds for every
$z\in\mathbb C^4$.
\end{theorem}

We begin with the following lemma which computes the number $A(s,t)$ in (\ref{a_st})
in some special cases. Throughout this note, we use the notation $\mathbb R_+=[0,\infty)$.

\begin{lemma}\label{A}
For $s,t\in\mathbb R_+^4$, we have the following:
\begin{enumerate}
\item[(i)] if $s_1t_2t_3s_4=t_1s_2s_3t_4$, then $A(s,t)=\sum_{i=1}^4 \sqrt{s_it_i}$;
\item[(ii)] if $s_1t_1=s_2t_2=s_3t_3=s_4t_4$, then $A(s,t)=2\sqrt[4]{s_1t_2t_3s_4}+2\sqrt[4]{t_1s_2s_3t_4}$.
\end{enumerate}
\end{lemma}

\begin{proof}
(i) We have
$$
A(s,t)= \inf_{r>0} \left(\sqrt{s_1s_4
r^{-2}+t_1t_4r^2+s_1t_1+s_4t_4} +\sqrt{s_2s_3
r^{-2}+t_2t_3r^2+s_2t_2+s_3t_3}\right).
$$
We first consider the case of $s_i,t_i>0$.
The first term has the minimum
$$
\sqrt{2\sqrt{s_1s_4t_1t_4}+s_1t_1+s_4t_4}=\sqrt{s_1t_1}+\sqrt{s_4t_4}
$$
when $s_1s_4r^{-2}=t_1t_4r^2$. 
Similarly, the second term has the minimum
$\sqrt{s_2t_2}+\sqrt{s_3t_3}$ when $s_2s_3r^{-2}=t_2t_3r^2$. From
$s_1t_2t_3s_4=t_1s_2s_3t_4$, we see that both terms have their
minimums simultaneously. If one of $s_i$ or $t_i$ is zero, we may check directly.

(ii) We also have
$$
\begin{aligned}
A(s,t)
&= \inf_{r>0} \left(\sqrt{s_1s_4 r^{-2}+t_1t_4r^2+s_1t_1+s_4t_4} +\sqrt{s_2s_3 r^{-2}+t_2t_3r^2+s_2t_2+s_3t_3}\right)\\
&= \inf_{r>0} \left(\sqrt{s_1s_4 r^{-2}+t_1t_4r^2+2\sqrt{s_1t_1s_4t_4}} +\sqrt{s_2s_3 r^{-2}+t_2t_3r^2+2\sqrt{s_2t_2s_3t_3}}\right)\\
&= \inf_{r>0} \left(\sqrt{s_1s_4}r^{-1}+\sqrt{t_1t_4}r+\sqrt{s_2s_3}r^{-1}+\sqrt{t_2t_3}r\right).\\
\end{aligned}
$$
The sum of the first and the last has the minimum
$2\sqrt[4]{s_1t_2t_3s_4}$ when $\sqrt{s_1s_4}r^{-1}=\sqrt{t_2t_3}r$.
The sum of the middle two terms has the minimum
$2\sqrt[4]{t_1s_2s_3t_4}$ when $\sqrt{t_1t_4}r=\sqrt{s_2s_3}r^{-1}$.
From $s_1t_1=s_2t_2=s_3t_3=s_4t_4$, we see that they have their
minimums simultaneously.
\end{proof}

We define two subsets $\mathcal D_{\rm I}$ and $\mathcal D_{\rm II}$ of
$\mathbb R_+^4\times \mathbb R_+^4$ by
$$
\begin{aligned}
\mathcal D_{\rm I} &= \{ (s,t)\in \mathbb R_+^4\times \mathbb R_+^4 : s_1t_2t_3s_4=t_1s_2s_3t_4, A(s,t)=1\},\\
\mathcal D_{\rm II} &= \{ (s,t)\in \mathbb R_+^4\times \mathbb R_+^4 : s_1t_1=s_2t_2=s_3t_3=s_4t_4, A(s,t)=1\},
\end{aligned}
$$
and subsets $D_1,D_2,\dots,D_6$ by
$$
\begin{aligned}
\mathcal D_i &= \{ (s_i e_i, t_i e_i)\in \mathbb R_+^4\times \mathbb R_+^4 : s_it_i=1 \}, \quad i=1,2,3,4,\\
\mathcal D_5 &= \{ (s_1,0,0,s_4),(0,t_2,t_3,0) \in \mathbb R_+^4\times \mathbb R_+^4: s_1t_2t_3s_4=1\slash 16\},\\
\mathcal D_6 &= \{ (0,s_2,s_3,0),(t_1,0,0,t_4)\in \mathbb R_+^4\times \mathbb R_+^4
: t_1s_2s_3t_4=1 \slash 16\}.
\end{aligned}
$$
Here, $\{e_1,e_2,e_3,e_4\}$ denotes the canonical basis of $\mathbb
R^4$. We see that $D_1,D_2,D_3,D_4\subset {\mathcal D_{\rm I}}$ and
$D_5,D_6\subset  {\mathcal D_{\rm II}}$ by Lemma \ref{A}.

\begin{lemma}\label{convex}
Every vector in $\mathcal D_{\rm I}$ is a convex combination of vectors in $\bigcup_{i=1}^4 \mathcal D_i$, and
every vector in $\mathcal D_{\rm II}$ is a convex combination of vectors in $\mathcal D_5 \cup \mathcal D_6$.
\end{lemma}

\begin{proof}
For $(s,t) \in \mathcal D_{\rm I}$, we have $\sum_{i=1}^4 \sqrt{s_it_i}=1$
by Lemma \ref{A} (i), and so the convex combination
$$
(s,t)=\sum_{i=1}^4\sqrt{s_it_i}\left(\sqrt{s_i \over t_i}e_i,
\sqrt{t_i \over s_i}e_i\right),\qquad \left(\sqrt{s_i \over t_i}e_i,
\sqrt{t_i \over s_i}e_i\right) \in \mathcal D_i
$$
is obtained. For $(s,t) \in \mathcal D_{\rm II}$, we have
$2\sqrt[4]{s_1t_2t_3s_4}+2\sqrt[4]{t_1s_2s_3t_4}=1$ by Lemma \ref{A}
(ii). We put
$$
\begin{aligned}
s'&={1 \over 2}\left( \sqrt[4]{s_1^3t_2^{-1}t_3^{-1}s_4^{-1}},0,0,\sqrt[4]{s_1^{-1}t_2^{-1}t_3^{-1}s_4^3}\right)\\
t'&={1 \over 2}\left( 0,\sqrt[4]{s_1^{-1}t_2^3t_3^{-1}s_4^{-1}},\sqrt[4]{s_1^{-1}t_2^{-1}t_3^3s_4^{-1}},0 \right)\\
s''&={1 \over 2}\left( 0,\sqrt[4]{t_1^{-1}s_2^3s_3^{-1}t_4^{-1}},\sqrt[4]{t_1^{-1}s_2^{-1}s_3^3t_4^{-1}},0 \right)\\
t''&={1 \over 2}\left(
\sqrt[4]{t_1^3s_2^{-1}s_3^{-1}t_4^{-1}},0,0,\sqrt[4]{t_1^{-1}s_2^{-1}s_3^{-1}t_4^3}\right),
\end{aligned}
$$
then we have the convex combination
$(s,t) = 2\sqrt[4]{s_1t_2t_3s_4} (s',t') +
2\sqrt[4]{t_1s_2s_3t_4}(s'',t'')$ with $(s',t')  \in \mathcal D_5$
and $(s'',t'') \in \mathcal D_6$.
\end{proof}

We recall that the {\sf X}-part of a three qubit entanglement
witness is again an entanglement witness by Proposition 3.1 of \cite{han_kye_GHZ},
and so we see that a three qubit {\sf X}-state $\varrho$ is separable if and
only if $\lan\varrho,W\ran := {\text{\rm tr}}(W\varrho^\ttt)\ge 0$ for every {\sf X}-shaped
entanglement witness $W=X(s,t,u)$. In this characterization, we may
assume that $A(s,t)=B(u)=1$ by the identity
$$
\begin{aligned}
X(s,t,u)=&B(u)X\left({1 \over A(s,t)}s,\ {1 \over A(s,t)}t,\ {1 \over B(u)}u\right)\\
  &\phantom{XXXXXXX}+X\left({A(s,t)-B(u) \over A(s,t)}s,\ {A(s,t)-B(u) \over A(s,t)}t,\ 0\right),
\end{aligned}
$$
since $A(\alpha s,\alpha t)=\alpha A(s,t)$ and $B(\alpha u)=\alpha B(u)$
for $\alpha >0$.
The next proposition explains the relations between two inequalities
$\lan\varrho,W\ran\ge 0$ and $\Delta_\varrho C(z)\ge \mathcal L
(\varrho, z)$ in various situations.

\begin{proposition}\label{witness-inequality}
For a three qubit {\sf X}-state $\varrho=X(a,b,c)$, we have the following:
\begin{enumerate}
\item[(i)]
for $i=1,2,3,4$, we have $\langle \varrho, W \rangle \ge 0$ for every {\sf X}-shaped witness $W=X(s,t,u)$
with $(s,t)\in {\mathcal D}_i$
if and only if $\sqrt{a_i b_i}C(z)\ge \mathcal L (\varrho, z)$ for
every $z \in \mathbb C^4$;
\item[(ii)]
$\langle \varrho, W \rangle \ge 0$ for every {\sf X}-shaped witness $W=X(s,t,u)$
with $(s,t)\in {\mathcal D}_5$
if and only if $\sqrt[4]{a_1b_2b_3a_4}C(z)\ge \mathcal L (\varrho, z)$ for every $z \in \mathbb C^4$;
\item[(iii)]
$\langle \varrho, W \rangle \ge 0$ for every {\sf X}-shaped witness $W=X(s,t,u)$
with $(s,t)\in {\mathcal D}_6$
if and only if $\sqrt[4]{b_1a_2a_3b_4}C(z)\ge\mathcal L (\varrho,
z)$ for every $z \in \mathbb C^4$;
\item[(iv)]
$\langle \varrho, W \rangle \ge 0$ for every {\sf X}-shaped witness $W=X(s,t,u)$
with $(s,t)\in {\mathcal D}_{\rm I}$
if and only if $\min_{1 \le i \le 4} \sqrt{a_i b_i}C(z)\ge \mathcal
L (\varrho, z)$ for every $z \in \mathbb C^4$;
\item[(v)]
$\langle \varrho, W \rangle \ge 0$ for every {\sf X}-shaped witness $W=X(s,t,u)$
with $(s,t)\in {\mathcal D}_{\rm II}$
if and only if $ \min \{ \sqrt[4]{a_1b_2b_3a_4},
\sqrt[4]{b_1a_2a_3b_4}\}C(z)\ge \mathcal L (\varrho, z)$ for every
$z \in \mathbb C^4$.
\end{enumerate}
\end{proposition}

\begin{proof}
We may assume that $A(s,t)=B(u)=1$ for an entanglement witness
$W=X(s,t,u)$.
For a given $z\in\mathbb C^4$, we will put $u=\dfrac
1{C(z)}(-z_1,-z_2,-z_3,-\bar z_4)$. Then we have $B(u)=1$ by (\ref{CB}), and
$$
\langle \varrho, W \rangle
= \sum_{j=1}^4 a_j s_j +\sum_{j=1}^4 b_j t_j - 2 {\mathcal L (\varrho, z) \over C(z)}.
$$
Therefore, we have $\langle \varrho, W \rangle \ge 0$ if and only if
$(\sum_{j=1}^4 a_j s_j +\sum_{j=1}^4 b_j t_j) C(z) \ge 2 {\mathcal L
(\varrho, z)}$. By Lemma \ref{A}, the statements (i), (ii) and (iii)
follow from the identities
$$
\begin{aligned}
\inf &\left\{\sum_{j=1}^4 a_j s_j +\sum_{j=1}^4 b_j t_j : (s,t) \in \mathcal D_i \right\}\\
=& \phantom{X}
\begin{cases}
\inf \{a_i s_i + b_i t_i : s_i t_i=1\} = 2\sqrt{a_ib_i},&\qquad 1 \le i \le 4 \\
\inf \{a_1s_1+b_2t_2+b_3t_3+a_4s_4 : s_1t_2t_3s_4={1 \over 16}\} = 2\sqrt[4]{a_1b_2b_3a_4},&\qquad i=5 \\
\inf \{b_1t_1+a_2s_2+a_3s_3+b_4t_4 : t_1s_2s_3t_4={1 \over 16}\} =
2\sqrt[4]{b_1a_2a_3b_4},&\qquad i=6.
\end{cases}
\end{aligned}
$$
The statements (iv) and (v) follows from (i), (ii), (iii) and Lemma \ref{convex}.
\end{proof}

For $\alpha,\beta,\gamma>0$, we define $D_{\alpha,\beta,\gamma}$ by
$$
D_{\alpha,\beta,\gamma} =
\diag (\alpha^{1\slash2},\alpha^{-1\slash2})\otimes
\diag (\beta^{1\slash2},\beta^{-1\slash2})\otimes
\diag (\gamma^{1\slash2},\gamma^{-1\slash2})
\in M_2\otimes M_2\otimes M_2.
$$

\begin{lemma}\label{su}
For every three qubit {\sf X}-shaped self-adjoint matrix $X(a,b,c)$ and $\a,\b,\g>0$, we have the following:
\begin{enumerate}
\item[(i)] $D_{\alpha,\beta,\gamma} X(a,b,c) D_{\alpha,\beta,\gamma}^*$ is again {\sf X}-shaped
of the form $X(x,y,c)$. The vectors $x,y$ satisfy the relations
$x_iy_i=a_ib_i$, $x_1y_2y_3x_4=a_1b_2b_3a_4$ and
$y_1x_2x_3y_4=b_1a_2a_3b_4$;
\item[(ii)]
if $a_1b_2b_3a_4=b_1a_2a_3b_4$, then there exist unique
$\alpha,\beta,\gamma>0$ so that $x_i=y_i=\sqrt{a_ib_i}$ for each
$i=1,2,3,4$;
\item[(iii)]
if $a_1b_1=a_2b_2=a_3b_3=a_4b_4$, then there exist unique
$\alpha,\beta,\gamma>0$ so that
$x_1=y_2=y_3=x_4=\sqrt[4]{a_1b_2b_3a_4}$ and
$y_1=x_2=x_3=y_4=\sqrt[4]{b_1a_2a_3b_4}$.
\end{enumerate}
\end{lemma}

\begin{proof}
We check directly that $D_{\alpha,\beta,\gamma} \varrho
D_{\alpha,\beta,\gamma}^*=X(x,y,c)$ with
\begin{equation}\label{hhjkljklhkjg}
\begin{aligned}
x&=(\alpha\beta\gamma a_1,\ \alpha\beta\gamma^{-1} a_2,\ \alpha\beta^{-1}\gamma a_3,\ \alpha\beta^{-1}\gamma^{-1}a_4),\\
y&=(\alpha^{-1}\beta^{-1}\gamma^{-1}b_1,\
\alpha^{-1}\beta^{-1}\gamma b_2,\
    \alpha^{-1}\beta\gamma^{-1} b_3,\ \alpha^{-1}\beta\gamma b_4),
\end{aligned}
\end{equation}
from which the required relations follow. For the statement (ii), we
put
$$
\alpha^4:={b_1b_4 \over a_1a_4}={b_2b_3 \over a_2a_3}, \qquad
\beta^4:={b_1a_3 \over a_1b_3}={b_2a_4 \over a_2b_4}, \qquad
\gamma^4:={b_1a_2 \over a_1b_2}={b_3a_4 \over a_3b_4}.
$$
Then we have
$$
x_1^4=\alpha^4\beta^4\gamma^4a_1^4={b_1b_4 \over
a_1a_4}\cdot{b_2a_4 \over a_2 b_4}\cdot{b_1a_2 \over a_1
b_2}\cdot a_1^4 = a_1^2b_1^2,
$$
and $x_i^4=a_i^2b_i^2$ for $i=2,3,4$, similarly.
From the relation $x_iy_i=a_ib_i$ in (i), we also see that
$y_i^4=a_i^2b_i^2$. For the uniqueness, we multiply
two relations $x_1=\alpha\beta\gamma a_1$ and $x_4=\alpha\beta^{-1}\gamma^{-1}a_4$ from
(\ref{hhjkljklhkjg}).
For (iii), we put
$$
\alpha^4:={b_2b_3 \over a_1a_4}, \qquad \beta^4:={b_2a_4 \over
a_1b_3}, \qquad \gamma^4:={b_3a_4 \over a_1b_2}.
$$
Then we have $x_1^4=y_2^4=y_3^4=x_4^4=a_1b_2b_3a_4$ by a direct computation.
From the relation $x_i^4y_i^4=a_i^4b_i^4=a_1b_1a_2b_2a_3b_3a_4b_4$,
we also have $y_1=x_2=x_3=y_4=\sqrt[4]{b_1a_2a_3b_4}$.
Uniqueness also follows in a similar way as in (ii).
\end{proof}

The following lemma is the main part of the proof, which shows that
the inequality {\rm (\ref{kye-main-ineq-sep})} implies the
separability when the diagonal entries of $\varrho=X(a,b,c)$ satisfy
some identities.

\begin{lemma}\label{average}
Suppose that $a_1b_2b_3a_4=b_1a_2a_3b_4$ or
$a_1b_1=a_2b_2=a_3b_3=a_4b_4$. Then a three qubit {\sf X}-state
$\varrho=X(a,b,c)$ is separable if and only if the inequality {\rm
(\ref{kye-main-ineq-sep})} holds for every $z\in\mathbb C^4$.
\end{lemma}

\begin{proof}
It suffices to prove the sufficiency.
Recall that $\vr$ is separable if and only if $\lan \varrho, W \ran
\ge 0$ for every {\sf X}-shaped entanglement witness $W$. So, we
assume the inequality {\rm (\ref{kye-main-ineq-sep})} for every
$z\in\mathbb C^4$ and will show that $\lan \varrho, W \ran \ge 0$
for every {\sf X}-shaped entanglement witness $W$. We prove the
assertion separately.

[Case I: $a_1b_2b_3a_4=b_1a_2a_3b_4$]: For a three qubit {\sf X}-shaped
self-adjoint matrix $X=X(s,t,u)$, we define
$$
X^\sym=X\left({s+t \over 2}, {s+t \over 2}, u\right)={X + (UXU^*)^\ttt
\over 2},
$$
with the local unitary $U=S\ot S\ot S$ and $S=\left(\begin{matrix}0&1\\1&0\end{matrix}\right)$.
Note that $X^\sym$ has {\sl symmetric} diagonal entries and shares the
anti-diagonal part with $X$. Furthermore, we have $X^\sym=X$ if and only if $s=t$. We apply Lemma \ref{su} (ii) to take $\a,\b,\g>0$ so
that $\varrho_\syms:=D_{\alpha,\beta,\gamma} \varrho
D_{\alpha,\beta,\gamma}^*$ has symmetric diagonals, that is,
satisfies $(\varrho_\syms)^\sym=\varrho_\syms$. Let
$W$ be an arbitrary three qubit {\sf X}-shaped entanglement witness. With notation
$D:=D_{\alpha,\beta,\gamma}^{-1}=D_{\alpha^{-1},\beta^{-1},\gamma^{-1}}$,
we have
$$
\begin{aligned}
\langle \varrho, W \rangle
= \langle D \varrho_\syms D^*, W \rangle = \langle \varrho_\syms, DWD^* \rangle
&= \langle (\varrho_\syms)^\sym, DWD^* \rangle \\
&= \left\langle {\varrho_\syms + (U\varrho_\syms U^*)^\ttt \over 2},
DWD^* \right\rangle.
\end{aligned}
$$
It follows that
$$
\begin{aligned}
\langle \varrho, W \rangle
=& {1 \over 2}\langle \varrho_\syms, DWD^* \rangle +{1 \over 2}\langle U\varrho_\syms U^*, (DWD^*)^\ttt \rangle \\
=& {1 \over 2}\langle \varrho_\syms, DWD^* \rangle +{1 \over 2}\langle \varrho_\syms, (U(DWD^*)U^*)^\ttt \rangle \\
=& \langle \varrho_\syms, (DWD^*)^\sym \rangle = \langle \varrho,
D^{-1}(DWD^*)^\sym (D^{-1})^* \rangle.
\end{aligned}
$$
Write $D^{-1}(DWD^*)^\sym (D^{-1})^*=X(x,y,u)$ by Lemma \ref{su} (i).
Since $(DWD^*)^\sym$ has symmetric diagonal entries,
$x$ and $y$ satisfy the relation $x_1y_2y_3x_4=y_1x_2x_3y_4$ by
Lemma \ref{su} (i) again. Furthermore, our assumption implies
$\Delta_\varrho=\min_{1 \le i \le 4} \sqrt{a_ib_i}$ by the relation
$(a_1b_2b_3a_4)(b_1a_2a_3b_4)=\Pi_{i=1}^4 a_ib_i$. Therefore, the
inequalities in the right side of Proposition
\ref{witness-inequality} (iv) hold, and we have finally $\langle
\varrho, W \rangle=\langle \varrho, X(x,y,u) \rangle\ge0$ by
Proposition \ref{witness-inequality} (iv). Therefore, we conclude
that $\varrho$ is separable.

[Case II: $a_1b_1=a_2b_2=a_3b_3=a_4b_4$]: In this case, our
assumption implies $\Delta_\varrho=\min \{ \sqrt[4]{a_1b_2b_3a_4},
\sqrt[4]{b_1a_2a_3b_4}\}$ similarly, and so we have the inequalities
in the right side of Proposition \ref{witness-inequality} (v). We consider the symmetric unitaries
$$
U_0=I\ot I\ot I,\quad
U_1=I\ot S\ot S,\quad
U_2=S\ot I\ot S,\quad
U_3=S\ot S\ot I
$$
with the identity $I$ and $S$ as above, and the partial transposes
$$
\tau_0=\tau_{\emptyset},\ \tau_1=\tau_{BC},\ \tau_2=\tau_{CA},\
\tau_3=\tau_{AB}, \quad \tau_0'=\tau_{ABC},\ \tau_1'=\tau_A,\
\tau_2'=\tau_B,\ \tau_3'=\tau_C
$$
on $M_A \otimes M_B \otimes M_C$ with $M_A = M_B = M_C = M_2$.
For a given three qubit {\sf X}-shaped self-adjoint matrix $X=X(s,t,u)$, we
define the matrix $X^\uni$  by
$$
X^\uni := {1 \over 4} \sum_{i=0}^3 (U_i X U_i^*)^{\tau_i} = X(s',t',u),
$$
with the two families of {\sl uniform} diagonal entries
$$
s_1'=t_2'=t_3'=s_4'= {1 \over 4}(s_1+t_2+t_3+s_4), \quad
t_1'=s_2'=s_3'=t_4'= {1 \over 4}(t_1+s_2+s_3+t_4).
$$
Note that $X^\uni$ shares the anti-diagonal part with $X$. We also have
$X^\uni=X$ if and only if $s_1=t_2=t_3=s_4$ and $t_1=s_2=s_3=t_4$.

We apply Lemma \ref{su} (iii) to take $\alpha,\beta,\gamma$ so that
$\varrho_\unis:=D_{\alpha,\beta,\gamma} \varrho
D_{\alpha,\beta,\gamma}^*$ has two families of uniform diagonal entries, that is,
satisfies $(\varrho_\unis)^\uni=\varrho_\unis$. Then we have
$$
\begin{aligned}
\langle \varrho, W \rangle = \langle D \varrho_\unis D^*, W \rangle &=
\langle \varrho_\unis, DWD^* \rangle \\
&= \langle (\varrho_\unis)^\uni, DWD^* \rangle = \left\langle {1 \over 4}
\sum_{i=0}^3 (U_i \varrho_\unis U_i^*)^{\tau_i}, DWD^* \right\rangle.
\end{aligned}
$$
One can verify that the relations
$(U_i P U_i^*)^{\tau_i'}=U_i P^{\tau_i'} U_i^*$ and
$\tr(P^{\tau_i'}Q)=\tr(PQ^{\tau_i'})$
hold for product matrices $P$ and $Q$, hence for any tensors $P$ and
$Q$. Therefore, we have
$$
\begin{aligned}
4\langle \varrho, W \rangle =& \sum_{i=0}^3 \tr \left((U_i \varrho_\unis
U_i^*)^{\tau_i} (DWD^*)^\ttt \right)
=  \sum_{i=0}^3 \tr \left((U_i \varrho_\unis U_i^*)^{\tau_i'} (DWD^*) \right) \\
=&  \sum_{i=0}^3 \tr \left((U_i \varrho_\unis U_i^*) (DWD^*)^{\tau'_i}
\right)
= \sum_{i=0}^3 \tr \left(\varrho_\unis (U_i(DWD^*)^{\tau'_i}U_i^*) \right) \\
=& \sum_{i=0}^3 \tr \left(\varrho_\unis (U_i(DWD^*)U_i^*)^{\tau_i'}
\right) =  \sum_{i=0}^3  \langle \varrho_\unis,
(U_i(DWD^*)U_i^*)^{\tau_i} \rangle,
\end{aligned}
$$
which implies
$\langle \varrho, W \rangle = \langle \varrho_\unis, (DWD^*)^\uni \rangle =
\langle \varrho, D^{-1}(DWD^*)^\uni (D^{-1})^* \rangle$.
Now, we write $D^{-1}(DWD^*)^\uni (D^{-1})^*=X(x,y,u)$. Since
$(DWD^*)^\uni$ has two families of uniform diagonal entries, we have
$x_1y_1=x_2y_2=x_3y_3=x_4y_4$ by Lemma \ref{su} (i). Therefore, we
have $\langle \varrho, W \rangle = \langle \varrho, D^{-1}(DWD^*)^\uni
(D^{-1})^* \rangle\ge 0$ by Proposition \ref{witness-inequality}
(v). This completes the proof.
\end{proof}

Now, we are ready to prove Theorem \ref{diagonal}. Suppose that
$\varrho$ is a three qubit {\sf X}-state $X(a,b,c)$ satisfying the inequality
{\rm (\ref{kye-main-ineq-sep})} for every $z \in \mathbb C^4$.
Applying Proposition \ref{witness-inequality} to the witness
$W=X\left({\sqrt{b_1 \over a_1} e_1, \sqrt{a_1 \over b_1}
e_1,-{|c_j| \over c_j} e_j}\right)$, we have $|c_j| \le\sqrt{a_1b_1}$
for $j=1,2,3,4$. If $b_2b_3a_4=0$ then the inequality {\rm
(\ref{kye-main-ineq-sep})} implies that $\vr$ is diagonal, and so
$\vr$ is separable. Let $b_2b_3a_4\ne0$. We may assume without loss
of generality that $\min_{1 \le i \le 4}
\sqrt{a_ib_i}=\sqrt{a_1b_1}$ and $\sqrt[4]{a_1b_2b_3a_4} \le
\sqrt[4]{b_1a_2a_3b_4}$ by the symmetry. We consider two cases
separately.

We first consider the case
$\Delta_\varrho=\min\{\sqrt[4]{a_1b_2b_3a_4},\sqrt[4]{b_1a_2a_3b_4}\}$.
In this case, we have $\Delta_\varrho = \sqrt[4]{a_1b_2b_3a_4}\le
\sqrt{a_1b_1}$. Put
$$
\lambda_2:={a_2b_2-a_1b_1 \over b_2}, \qquad
\lambda_3:={a_3b_3-a_1b_1 \over b_3}, \qquad \mu_4:={a_4b_4-a_1b_1
\over a_4},
$$
which are nonnegative. Then we have $(a_2-\lambda_2)b_2=a_1b_1$,
$(a_3-\lambda_3)b_3=a_1b_1$ and $a_4(b_4-\mu_4)=a_1b_1$, and so
$$
b_1(a_2-\lambda_2)(a_3-\lambda_3)(b_4-\mu_4)
=b_1 \cdot{a_1b_1 \over b_2}\cdot{a_1b_1 \over b_3}\cdot{a_1b_1 \over a_4}
={(a_1b_1)^4 \over a_1b_2b_3a_4}
\ge a_1b_2b_3a_4.
$$
We put $a'=(a_1,a_2-\lambda_2,a_3-\lambda_3,a_4)$ and
$b'=(b_1,b_2,b_3,b_4-\mu_4)$. Then $\varrho'=X(a',b',c)$ is an {\sf
X}-state satisfying $a_i'b_i'=a_1b_1$  for $i=1,2,3,4$ and
$\Delta_{\varrho'} = \sqrt[4]{a_1b_2b_3a_4}= \Delta_\varrho$.
Therefore, we have $\mathcal L (\varrho', z) = \mathcal L (\varrho,
z) \le C(z) \Delta_\varrho = C(z) \Delta_{\varrho'}$ for every $z
\in \mathbb C^4$. By the relation $a_i'b_i'=a_1b_1$ and Lemma
\ref{average}, we see that $\varrho'$ is separable. Therefore, we
can conclude that $\varrho=\varrho'+{\rm diag}(0,\lambda_2,\lambda_3,0,\mu_4,0,0,0)$
is separable.

Next, we consider the case $\Delta_\varrho=\min_{1 \le i \le 4}
\sqrt{a_ib_i}$. Define the continuous functions $f_2,f_3,f_4,f :
[0,1] \to \mathbb R$ by
$$
\begin{aligned}
f_2(t)&=(1-t)a_2 + t{a_1b_1 \over b_2}, \qquad &f_3(t)&=(1-t)a_3 + t{a_1b_1 \over b_3}\\
f_4(t)&=(1-t)b_4 + t{a_1b_1 \over a_4}, \qquad &f(t)&=b_1 f_2(t)
f_3(t) f_4(t).
\end{aligned}
$$
The assumption $\Delta_\varrho=\sqrt{a_1b_1}$ implies
$$
\frac{a_1b_1}{b_2}\le f_2(t) \le a_2, \qquad
\frac{a_1b_1}{b_3}\le f_3(t) \le a_3, \qquad
\frac{a_1b_1}{a_4}\le f_4(t) \le b_4
$$
and
$f(0)=b_1a_2a_3b_4 \ge a_1b_2b_3a_4 \ge {(a_1b_1)^4 \over
a_1b_2b_3a_4} = f(1)$.
By the intermediate value theorem, there exists $t_0\in[0,1]$ such
that $f(t_0)=a_1b_2b_3a_4$. We put
$a'=(a_1,f_2(t_0),f_3(t_0),a_4)$ and $b'=(b_1,b_2,b_3,f_4(t_0))$.
Then, $\varrho'=X(a',b',c)$ is an {\sf X}-state satisfying
$a_1'b_2'b_3'a_4'=b_1'a_2'a_3'b_4'$ and $\Delta_{\varrho'} =
\sqrt{a_1b_1}= \Delta_\varrho$. Therefore, we have
$\mathcal L (\varrho', z) = \mathcal L (\varrho, z) \le C(z)
\Delta_\varrho = C(z) \Delta_{\varrho'}$
for every $z \in \mathbb C^4$. By the identity $a_1'b_2'b_3'a_4'=b_1'a_2'a_3'b_4'$ and Lemma
\ref{average}, we  see that $\varrho'$ is separable. By the relation
$$
\varrho=\varrho'+{\rm
diag}(0,a_2-f_2(t_0),a_3-f_3(t_0),0,b_4-f_4(t_0),0,0,0),
$$
we conclude that $\varrho$ is separable. This
completes the proof of Theorem \ref{diagonal}.


\section{Norms arising from separability problem}\label{kye-section-norm}

In this section, we introduce the norm $\|\ \|_\xx$ 
and its dual norm $\|\ \|_\xx^\prime$ on the four-dimensional space $\mathbb C^4$,
and discuss how these norms are related to
the separability problem. We also investigate their elementary properties.
Theorem \ref{diagonal} shows that a three qubit {\sf X}-state
$\varrho=X(a,b,c)$ is separable if and only if the inequality
\bea\label{kye_main}
\Delta_\varrho
\ge \max_{z\in\mathbb C^4} \frac {\re (c_1z_1 + c_2z_2+ c_3z_3 +
\bar c_4 z_4)} {\max_\sigma \left(|z_1 e^{{\rm i}\sigma}+z_4| + |z_2
e^{{\rm i}\sigma} +\bar z_3|\right)  } \eea holds. Since $\re (\bar
c_4 z_4)=\re (c_4\bar z_4)$, we see that the right hand side is
equal to
$$
\max_{z\in\mathbb C^4} \frac {\re (c_1z_1 + c_2z_2+ c_3z_3 +  c_4
\bar z_4)} {\max_\sigma \left(|z_1 e^{{\rm i}\sigma}+z_4| + |z_2
e^{{\rm i}\sigma} +\bar z_3|\right)  } = \max_{z\in\mathbb C^4}
\frac {\re (c_1z_1 + c_2z_2+ c_3z_3 +  c_4 z_4)} {\max_\sigma
\left(|z_1 e^{{\rm i}\sigma}+\bar z_4| + |z_2 e^{{\rm i}\sigma}
+\bar z_3|\right)  }.
$$
Therefore, it is natural to consider $\|z\|_\xx$, as it was defined
by (\ref{kye_norm-def}).
We write $z_i= s_ie^{{\rm i}{\sigma_i}}$
with $s_i\ge0$. From the relation
\begin{equation}\label{eq:zx}
\begin{aligned}
\| z\|_\xx &=\max_\sigma \left(|s_1 e^{{\rm
i}(\sigma_1+\sigma_4+\sigma)}+s_4|
   + |s_2 e^{{\rm i}(\sigma_2+\sigma_3+\sigma)} +s_3|\right)\\
&=\max_\sigma \left(|s_1 e^{{\rm i}(\phi_z+\sigma)}+s_4|
   + |s_2 e^{{\rm i}\sigma} +s_3|\right)\\
&=\max_\sigma \left(|s_1 e^{{\rm i}(-\phi_z-\sigma)}+s_4|
   + |s_2 e^{{\rm i}(-\sigma)} +s_3|\right)),
\end{aligned}
\end{equation}
we have
$$
\| z\|_\xx=\| (s_1 e^{{\rm i}{\phi_z}}, s_2,s_3,s_4)\|_\xx=\| (s_1
e^{-{\rm i}{\phi_z}}, s_2,s_3,s_4)\|_\xx.
$$
Therefore, we have the following:
\begin{equation}\label{kye-norm-phase=basic}
|z_i|=|w_i|, \ |\phi_z|=|\phi_w|\ \Longrightarrow\
\|z\|_\xx=\|w\|_\xx,
\end{equation}
for $z,w\in\mathbb C^4$. Because $\phi_z=\phi_{\alpha z}$ for a
complex scalar $\alpha$, we have $\|\alpha
z\|_\xx=|\alpha|\|z\|_\xx$. This equality can be seen also directly from the definition.
The inequality $\|z+w\|_\xx\le
\|z\|_\xx+\|w\|_\xx$ is clear, and so we see that $\|\ \|_\xx$ is
really a norm on the vector space $\mathbb C^4$.
We list up several
elementary properties of the norm $\|\ \|_\xx$ without proof. We use
the standard notations for norms; $\|z\|_\infty=\max_{1\le j\le 4}
\abs{z_j}$ and $\|z\|_1=\sum^4_{j=1}\abs{z_j}$.

\begin{proposition}\label{X-norm}
For $z\in\mathbb C^4$, we have the following:
\begin{enumerate}
\item[(i)]
$\|(z_1,z_2,z_3,z_4)\|_\xx=\|(\bar z_1,\bar z_2, \bar z_3,\bar z_4)\|_\xx$;
\item[(ii)]
if one entry of $z$ is zero then we have $\|z\|_\xx=\|
(|z_1|,|z_2|, |z_3|, |z_4|)\|_\xx=\|z\|_1$;
\item[(iii)]
if $\phi_z=0$ then $\|z\|_\xx=\| (|z_1|,|z_2|, |z_3|,
|z_4|)\|_\xx=\|z\|_1$.
\end{enumerate}
\end{proposition}

If $\phi_z=\pi$, then we can evaluate $\|z\|_\xx=\| (-|z_1|,|z_2|,
|z_3|, |z_4|)\|_\xx$ in terms of entries by a result in
\cite{han_kye_GHZ}. To explain this, we define
$$
\begin{aligned}
\Omega^-_i&=\left\{ z\in\mathbb C^4: \textstyle{ \frac
1{|z_i|} \ge \sum_{j\neq i}\frac 1{|z_j|}}\right\},\qquad i=1,2,3,4,\\
\Omega^-_0&=\mathbb
C^4\setminus\left(\sqcup_{i=1}^4\Omega^-_i\right)
=\left\{z\in\mathbb C^4: \textstyle{\frac 1{|z_i|} < \sum_{j\neq
i}\frac 1{|z_j|}},\ i=1,2,3,4\right\}.
\end{aligned}
$$
Then by Proposition 5.1 of \cite{han_kye_GHZ}, we have
\begin{equation}\label{X-norm-phase_pi}
\|z\|_\xx=
\begin{cases}
|z_1|+|z_2|+|z_3|+|z_4|-2|z_i|,\quad &z\in \Omega^-_i\ (i=1,2,3,4),\\
\Lambda(|z_1|,|z_2|,|z_3|,|z_4|),\quad & z\in \Omega^-_0,
\end{cases}
\end{equation}
whenever $\phi_z=\pi$, where
\bea
\label{eq:lambda}
\Lambda(a_1,a_2,a_3,a_4):=
\sqrt{\frac{(a_1a_2+a_3a_4)(a_1a_3+a_2a_4)(a_1a_4+a_2a_3)}{a_1a_2a_3a_4}}.
\eea

For $z=(re^{{\rm i}\phi},r,s,s)$, we can also compute the norm.
To do this, we note
$$
\begin{aligned}
\|(re^{{\rm i}\phi},r,s,s)\|_{\sf X}
&=\max_\theta |re^{{\rm i}\phi} e^{{\rm i}\theta}+s|+|re^{{\rm i}\theta}+s|\\
&=\max_\theta |re^{{\rm i}\phi} + se^{{\rm i}\theta}|+|r+se^{{\rm i}\theta}|.
\end{aligned}
$$
We fix two points $A(-r)$ and $B(-re^{{\rm i}\phi})$ on the complex
plane, and move around the point $P(se^{{\rm i}\theta})$ along the
circle with the radius $s$. Then the maximum in the above formula is
taken when the three points $A$, $B$ and $P$ make an isosceles
triangle, or equivalently when $P$ is located at $s e^{{\rm
i}\theta}$ or $-s e^{{\rm i}\theta}$ with $\theta=\frac 12
(\pi+(\phi+\pi))=\phi/2+\pi$. See Figure 1. Therefore, we have
\begin{equation}\label{norm-two-two}
\begin{aligned}
\|(re^{{\rm i}\phi},r,s,s)\|_{\sf X}
&= 2\max_\pm |r \pm se^{{\rm i}(\phi/2+\pi)}|\\
&= 2\max_\pm |r \pm se^{{\rm i}{\phi/2}}|
= 2\sqrt{r^2+s^2+2rs |\cos (\phi/2)|}.
\end{aligned}
\end{equation}
Especially, if the entries share a common magnitude then we have
\begin{equation}\label{X-norm-share-mag}
|z_i|=r\ (i=1,2,3,4)\ \Longrightarrow\ \|z\|_\xx= 2\sqrt 2\,r
\sqrt{1+|\cos(\phi/2)|}.
\end{equation}

\setlength{\unitlength}{1 mm}
\begin{figure}
\setlength{\unitlength}{1 mm}
\begin{center}
\begin{picture}(50,50)
  \qbezier(25.000,0.000)(33.284,0.000)
          (39.142,5.858)
  \qbezier(39.142,5.858)(45.000,11.716)
          (45.000,20.000)
  \qbezier(45.000,20.000)(45.000,28.284)
          (39.142,34.142)
  \qbezier(39.142,34.142)(33.284,40.000)
          (25.000,40.000)
  \qbezier(25.000,40.000)(16.716,40.000)
          (10.858,34.142)
  \qbezier(10.858,34.142)( 5.000,28.284)
          ( 5.000,20.000)
  \qbezier( 5.000,20.000)( 5.000,11.716)
          (10.858,5.858)
  \qbezier(10.858,5.858)(16.716,0.000)
          (25.000,0.000)
\put(25,20){\circle*{1}}
\put(10,20){\circle*{1}}
\put(17.5,7.01){\circle*{1}}
\put(42.31,30){\circle*{1}}
\drawline(10,20)(17.5,7.01)
\drawline(17.5,7.01)(42.31,30)
\drawline(42.31,30)(10,20)
\drawline(42.31,30)(13.75,13.5)
\drawline(10,20)(40.2,7.01)
\drawline(17.5,7.01)(40.2,7.01)
\put(6,20){$A$}
\put(16.5,3.01){$B$}
\put(43.31,31){$P$}
\end{picture}
\end{center}
\caption{Suppose that $A$ and $B$ are two fixed points with the same distances from the origin,
and $P$ moves around a circle centered at the origin. Then, the length
$\overline{AP}+\overline{PB}$ takes the maximum when the three
points $A,B$ and $P$ make an isosceles triangle.
}
\end{figure}
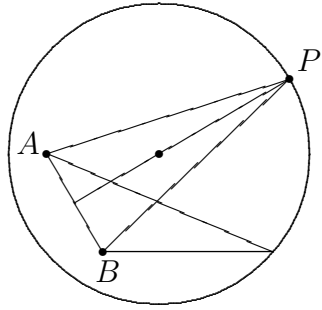

With the bi-linear form $\lan c,z\ran=\sum_{i=1}^4c_iz_i$, we see
that the right-hand side of \eqref{kye_main} becomes the dual norm
$\|c\|_\xx^\prime$ as follows:
$$
\begin{aligned}
\|c\|_\xx^\prime
&=\max_z \frac {\re \lan c,z\ran}{\| z\|_\xx}
=\max_z \frac {|\lan c,z\ran|}{\| z\|_\xx}\\
&=\max\{\re \lan c,z\ran: \|z\|_\xx=1\} =\max\{|\lan c,z\ran|:
\|z\|_\xx=1\}.
\end{aligned}
$$
By Theorem \ref{diagonal}, we have the following:

\begin{theorem}\label{kye-main-dual-norm}
A three qubit {\sf X}-state $\varrho=X(a,b,c)$ is separable if and
only if the inequality $\Delta_\varrho\ge \|c\|^\prime_\xx$ holds.
\end{theorem}

It is also possible to describe the dual norm $\|\ \|_\xx^\prime$ in
terms of separability. For given $t>0$ and $c\in\mathbb C^4$,
Theorem \ref{kye-main-dual-norm} tells us that the state $X(t{\bf
1},t{\bf 1}, c)$ is separable if and only if $t\ge
\|c\|_\xx^\prime$, where ${\bf 1}=(1,1,1,1)$. We
denote by $\mathbb S$ the convex cone of all unnormalized separable
states. Then, we have
\begin{equation}\label{dual-sep}
\|c\|_\xx^\prime=\inf\{t>0: X\left(t{\bf 1},t{\bf 1}, c\right)\in\mathbb S\}.
\end{equation}
This formula tells us that the dual norm $\|c\|_\xx^\prime$ is nothing
but the order unit norm of $X(0,0,c)$ in the ordered $*$-vector space $M_2 \otimes M_2 \otimes
M_2$ equipped with the positive cone $\mathbb S$ and the order unit
$\bf 1$. For the general information on the order unit and the order unit norm, we
refer to \cite{A} and \cite{PT}. The formula (\ref{dual-sep}) is also useful to investigate the
properties of the dual norm. For example, we will use the notion of
separability to prove that the dual norm is invariant under the same
phase difference.

\begin{proposition}\label{pp:state}
Suppose that $c,d\in\mathbb C^4$ satisfy $|c_i|=|d_i|$ for each
$i=1,2,3,4$ and $|\phi_c|=|\phi_d|$. Then $X(a,b,c)$ is separable if
and only if $X(a,b,d)$ is separable. Especially, if one of
anti-diagonals of an {\sf X}-state is zero, then the remaining three
phases are irrelevant to the separability.
\end{proposition}
\begin{proof}
Write $c_j=\abs{c_j}e^{{\rm i}\t_j}$ for $j=1,2,3,4$. Define the
product unitary matrix $P$ by
$$
P=\diag\left(1,e^{{\rm i}{\t_2+\t_3\over2}}\right)
\ox\diag\left(1,e^{{\rm i}{\t_1-\t_3\over2}}\right)\ox
\diag\left(1,e^{{\rm i}{\t_1-\t_2\over2}}\right).
$$
One can verify that $P X(a,b,c) P^*=X(a,b,c')$ where
$c'=(\abs{c_1},\abs{c_2},\abs{c_3},\abs{c_4}e^{{\rm i}\phi_c})$. Similarly,
we can find a product unitary matrix $Q$ such that $Q X(a,b,d)
Q^*=X(a,b,d')$, where
$d'=(\abs{d_1},\abs{d_2},\abs{d_3},\abs{d_4}e^{{\rm i}\phi_d})$. If
$\phi_c=\phi_d$ then $X(a,b,c')=X(a,b,d')$. If $\phi_c=-\phi_d$ then
$X(a,b,c')=X(a,b,d')^\ttt$. In either case, $X(a,b,c)$ is
separable if and only if $X(a,b,d)$ is separable, because the
separability is unchanged under product unitary transformation.

It remains to prove the last assertion. If one of anti-diagonals of
an \sx state is zero, one can find another product unitary matrix
$R$, such that the other three anti-diagonals of $RX(a,b,c)R^*$ are
nonnegative and real. So, their phases are irrelevant to the
separability.
\end{proof}

\begin{proposition}\label{dual-norm}
For $c,d\in\mathbb C^4$, we have the following:
\begin{enumerate}
\item[(i)]
if $|c_i|=|d_i|$ for each $i=1,2,3,4$ and $|\phi_c|=|\phi_d|$  then
$\|c\|^\prime_\xx=\| d\|^\prime_\xx$;
\item[(ii)]
$\|(c_1,c_2,c_3,c_4)\|^\prime_\xx=\|\bar c_1,\bar c_2, \bar c_3,\bar
c_4)\|^\prime_\xx$;
\item[(iii)]
if one of entries of $c$ is zero then we have
$\|c\|_\xx^\prime=\|(|c_1|, |c_2|, |c_3|, |c_4|)\|_\xx^\prime$.
\end{enumerate}
\end{proposition}

\begin{proof}
The statements (i) and (iii) follow from Proposition \ref{pp:state}
and (\ref{dual-sep}).
The identity (ii) is a special case of (i).
\end{proof}

We mention here that Proposition \ref{dual-norm} (i) can be proved directly
without using Proposition \ref{pp:state}.
We now consider the question which permutations of entries preserve the
norms $\|\ \|_\xx$ and $\|\ \|_\xx^\prime$.
We will use the notation $\sigma=\lan \sigma(1)\sigma(2)\sigma(3)\sigma(4)\ran$ for a
permutation $\sigma$ on $\{1,2,3,4\}$. For example, $\lan 1234\ran$
will denote the identity permutation. For $z\in\mathbb C^4$, we will
define $z^\sigma=(z_{\sigma(1)}, z_{\sigma(2)}, z_{\sigma(3)},
z_{\sigma(4)})\in\mathbb C^4$. It is clear by definition that
$\|z^\sigma\|_\xx=\|z\|_\xx$ when $\sigma$ is one of the following
eight permutations:
\begin{equation}\label{permutation}
\lan 1234\ran,\ \lan 1324\ran,\ \lan 4231\ran,\ \lan 4321\ran,\ \lan
2143\ran,\ \lan 2413\ran,\ \lan 3142\ran,\ \lan 3412\ran.
\end{equation}

\begin{proposition}\label{kyeperm-char}
For a permutation $\sigma$ on $\{1,2,3,4\}$, the following are equivalent:
\begin{enumerate}
\item[(i)]
$\|z^\sigma\|_\xx=\|z\|_\xx$ for every $z\in\mathbb C^4$;
\item[(ii)]
$\|c^\sigma\|^\prime_\xx=\|c\|^\prime_\xx$ for every $c\in\mathbb C^4$;
\item[(iii)]
$|\phi_{z^\sigma}|=|\phi_z|$ for every $z\in\mathbb C^4$;
\item[(iv)]
$\sigma$ is one of the permutations listed in {\rm (\ref{permutation})}.
\end{enumerate}
\end{proposition}

\begin{proof}
We have already seen the implication (iv) $\Longrightarrow$ (i), and the
equivalence between (i) and (ii) follows by the duality.
Consider $w=(e^{{\rm i}\theta_1}, e^{{\rm i}\theta_2},e^{{\rm
i}\theta_3},e^{{\rm i}\theta_4})\in\mathbb C^4$. Then we have
$$
\|w^\sigma\|_\xx=2\sqrt 2\sqrt{1+|\cos (\phi_{w^\sigma}/2)|},\qquad
\|w\|_\xx=2\sqrt 2\sqrt{1+|\cos (\phi_{w}/2)|}
$$
by (\ref{X-norm-share-mag}), and so we also have
\begin{equation}\label{kye_temp_xx}
w=(e^{{\rm i}\theta_1}, e^{{\rm i}\theta_2},e^{{\rm
i}\theta_3},e^{{\rm i}\theta_4}),\
\|w^\sigma\|_\xx=\|w\|_\xx\ \Longrightarrow\ |\phi_{w^\sigma}|= |\phi_w|.
\end{equation}
Suppose that (i) holds. For a given $z\in\mathbb C^4$, we write $z_i=r_i e^{{\rm i}\theta_i}$,
and put $w=(e^{{\rm i}\theta_1}, e^{{\rm i}\theta_2},e^{{\rm
i}\theta_3},e^{{\rm i}\theta_4})$. Then $\|w^\sigma\|_\xx=\|w\|_\xx$ implies
$|\phi_{z^\sigma}|=|\phi_{w^\sigma}|= |\phi_w|=|\phi_z|$ by (\ref{kye_temp_xx}).
This proves the direction (i) $\Longrightarrow$ (iii).
If (iii) is true then we have the relation
$$
|(\theta_1+\theta_4)-(\theta_2+\theta_3)|=|(\theta_{\sigma(1)}+\theta_{\sigma(4)})-(\theta_{\sigma(2)}+\theta_{\sigma(3)})|
$$
holds for every $\theta_1,\theta_2,\theta_3$ and $\theta_4$.
Therefore, we see that either
$\{\sigma(1),\sigma(4)\}=\{1,4\}$ or
$\{\sigma(1),\sigma(4)\}=\{2,3\}$ must hold. This shows that
$\sigma$ is one of permutations listed in {\rm (\ref{permutation})}.
\end{proof}

It is clear by Proposition \ref{kyeperm-char} that the eight
permutations in (\ref{permutation}) make a group. It turns out that
this is isomorphic to the dihedral group $D_4$ of order eight. We note that $\lan
1234\ran$, $\lan 2143\ran$, $\lan 3412\ran$ and $\lan 4321\ran$ are
even permutations, and the others are odd. These four even
permutations reflect the fact that the separability is invariant
under partial transposes. Indeed, if we denote by $\Gamma_A, \Gamma_B$ and
$\Gamma_C$ the partial transposes with respect to the $A$, $B$ and $C$
systems, respectively, then we have
$$
\begin{aligned}
X(a,b,c)^{\Gamma_C}&=X(a,b,(c_2,c_1,c_4,c_3)),\\
X(a,b,c)^{\Gamma_B}&=X(a,b,(c_3,c_4,c_1,c_2)),\\
X(a,b,c)^{\Gamma_A}&=X(a,b,(\bar c_4,\bar c_3,\bar c_2,\bar c_1)).
\end{aligned}
$$
On the other hand, if we interchange the $B$ and $C$ systems then
$\varrho=X(a,b,c)$ becomes
$$
X((a_1,a_3,a_2,a_4),(b_1,b_3,b_2,b_4),(c_1,c_3,c_2,c_4)).
$$
This reflects the odd permutation $\lan 1324\ran$ in
(\ref{permutation}). The remaining permutations in
(\ref{permutation}) are composition of $\lan 1324\ran$ and the
others. If we interchange $A$-$C$ and $A$-$B$ systems, then the
state $\varrho=X(a,b,c)$ becomes
\begin{equation}\label{kye-system-change}
\begin{aligned}
&X((a_1,b_4,a_3,b_2),(b_1,a_4,b_3,a_2),(c_1,\bar c_4,c_3,\bar c_2)),\\
&X((a_1,a_2,b_4,b_3),(b_1,b_2,a_4,a_3),(c_1,c_2,\bar c_4,\bar c_3)),
\end{aligned}
\end{equation}
respectively. We note that  the phase differences are invariant in both case.

\begin{proposition}\label{fytygdgfnd}
We have the following identities:
\begin{enumerate}
\item[(i)]
$\|(c_1,c_2,c_3,c_4)\|_\xx^\prime=\|(c_1,\bar c_4, c_3,\bar c_2)\|_\xx^\prime=\|(c_1,c_2,\bar c_4,\bar c_3)\|_\xx^\prime$;
\item[(ii)]
$\|(c_1,c_2,c_3,c_4)\|_\xx=\|(c_1,\bar c_4, c_3,\bar c_2)\|_\xx=\|(c_1,c_2,\bar c_4,\bar c_3)\|_\xx$.
\end{enumerate}
\end{proposition}

\begin{proof}
The identities in (i) follow from
(\ref{kye-system-change}) and (\ref{dual-sep}). 
We denote $\tilde c=(c_1,\bar c_4, c_3,\bar c_2)$, and $\tilde z=(z_1,\bar z_4, z_3,\bar z_2)$.
Note that $\|\ \|_\xx$ is the dual norm of $\|\ \|_\xx^\prime$ by the duality. Therefore, we have
$$
\|\tilde c\|_\xx
=\sup\frac{{\rm Re}\, \lan \tilde c,\tilde z\ran}{\|\tilde z\|_\xx^\prime}
=\sup\frac{{\rm Re}\, \lan  c, z\ran}{\|\tilde z\|_\xx^\prime}
=\sup\frac{{\rm Re}\, \lan  c, z\ran}{\|z\|_\xx^\prime}
=\|c\|_\xx
$$
by (i). The same argument is applied for $(c_1,c_2,\bar c_4,\bar c_3)$.
\end{proof}

\begin{theorem}\label{kye_two-partition}
Let $z\in\mathbb C^4$. If there exists a partition
$\{i_1,i_2\}\cup\{i_3,i_4\}=\{1,2,3,4\}$ such that
$|z_{i_1}|=|z_{i_2}|=r$ and $|z_{i_3}|=|z_{i_4}|=s$, then
$\|z\|_\xx=2\sqrt{r^2+s^2+2rs |\cos (\phi_z/2)|}$.
\end{theorem}

\begin{proof}
We first consider the case  $|z_1|=|z_2|$ and $|z_3|=|z_4|$. This is just (\ref{norm-two-two}).
If $|z_1|=|z_3|$ and $|z_2|=|z_4|$ then the identity follows from the first case and
Proposition \ref{kyeperm-char} with $\sigma=\langle 1324\rangle$. Finally,
if $|z_1|=|z_4|$ and $|z_2|=|z_3|$ then we get the result by Proposition \ref{fytygdgfnd}.
\end{proof}

We close this section to see how the norm $\|\ \|_\xx$ depends on the phase difference.
To do this, we fix nonnegative $s_i$'s, and define
the function
$$
\beta: \phi\mapsto \|(s_1e^{{\rm
i}\phi},s_2,s_3,s_4)\|_\xx=\|(s_1,s_2,s_3,s_4 e^{{\rm
i}\phi})\|_\xx.
$$
Note that $\beta$ is an even function by
(\ref{kye-norm-phase=basic}).
The relation (\ref{norm-two-two})
suggests that the function $\beta$ might be decreasing on $[0,\pi]$. This is the
case, in general.

\begin{proposition}\label{beta_decrease}
The function $\b$ is strictly decreasing on the interval $\phi\in[0,\pi]$ if and only if $s_j>0$ for $j=1,2,3,4$.
\end{proposition}
\bpf
The "only if" part follows from the definition of $\|\ \|_\xx$, or Proposition \ref{X-norm} (ii).
To prove the "if" part, suppose $s_j>0$ for $j=1,2,3,4$.
We consider the following two variable function
$$
\begin{aligned}
B(\sigma,\phi)&=|s_1e^{{\rm i}(\sigma-\phi)} +s_4|+|s_2 e^{{\rm
i}\sigma} +  s_3| \\
&=\sqrt{s_1^2+s_4^2+2s_1s_4\cos(\s-\phi)} +
\sqrt{s_2^2+s_3^2+2s_2s_3\cos\s}.
\end{aligned}
$$
Then $\beta(\phi)=\max_\sigma B(\sigma,\phi)$. We assume that
$0\le\phi\le\pi$. Suppose that $0\le\sigma\le \pi$. Then we have
$$
\begin{aligned}
\cos(\sigma-\phi)
&=\cos\sigma\cos\phi+\sin\sigma\sin\phi\\
&\ge \cos\sigma\cos\phi-\sin\sigma\sin\phi
=\cos(-\sigma-\phi),
\end{aligned}
$$
and $\cos\sigma=\cos(-\sigma)$. So, we have $B(\sigma,\phi)\ge
B(-\sigma,\phi)$, and see that the maximum $\max_\sigma
B(\sigma,\phi)$ occurs when $\sigma\in [0,\pi]$.

We fix $\phi_0$ with $0<\phi_0<\pi$, and consider two functions
$B_1$ and $B_2$ by
$$
B_1(\sigma)=|s_1e^{{\rm i}(\sigma-\phi_0)} +s_4|, \qquad
B_2(\sigma)=|s_2 e^{{\rm i}\sigma} +  s_3|.
$$
Then both $B_1$ and $B_2$ are differentiable at $\sigma=\phi_0$ and $\sigma=0$,
respectively, with $B_1^\prime(\phi_0)=B_2^\prime(0)=0$. The
function $B_1$ is increasing on $[0,\phi_0]$ and decreasing on
$[\phi_0,\pi]$. $B_2$ is decreasing on $[0,\pi]$. Therefore, the
maximum of $B(\s,\phi_0)=B_1(\sigma)+B_2(\sigma)$ occurs at $\sigma_0$ with
$0<\sigma_0<\phi_0$. Now, suppose that $\sigma_0<\phi<\phi_0$. Then
we have
$$
-\pi< -\phi_0< \sigma_0-\phi_0 <\sigma_0-\phi <0.
$$
This implies that
$$
\begin{aligned}
\beta(\phi_0)
&=B(\sigma_0,\phi_0)\\
&=|s_1e^{{\rm i}(\sigma_0-\phi_0)} +s_4|+|s_2 e^{{\rm i}\sigma_0} +  s_3|\\
&< |s_1e^{{\rm i}(\sigma_0-\phi)} +s_4|+|s_2 e^{{\rm i}\sigma_0} +
s_3| =B(\sigma_0,\phi)\le\beta(\phi).
\end{aligned}
$$
In short, we conclude that for every $\phi_0\in(0,\pi)$ there
exists $\sigma_0\in (0,\phi_0)$ such that
$$
\sigma_0<\phi<\phi_0 \ \Longrightarrow \beta(\phi)>\beta(\phi_0).
$$
Suppose that $0<\phi_1<\phi_2<\pi$. Because $\beta$ is continuous,
it has the maximum on the compact interval $[\phi_1,\phi_2]$ which must be taken
at $\phi_1$ by the above conclusion. Therefore, we have $\beta(\phi_1)>\beta(\phi_2)$.
\epf

\begin{corollary}
Suppose that $|z_i|=|w_i|>0$ for every $i=1,2,3,4$. Then we have
$|\phi_z|=|\phi_w|$  if and only if $\|z\|_\xx=\| w\|_\xx$.
\end{corollary}


\section{Dual norm and separability criterion}
\label{sec:norm}

Theorem \ref{kye-main-dual-norm} tells us that the separability problem
reduces to computing the dual norm.
For example, we have $\|(1,0,0,0)\|_\xx^\prime=1$, and so
the {\sf X}-state $\varrho_{a,b,c}$ given by
\begin{equation}\label{acin-exam}
\varrho_{a,b,c}=X((1,a,b,c),(1,a^{-1},b^{-1},c^{-1}),(1,0,0,0))
\end{equation}
is separable if and only if $abc^{-1}\ge 1$ and $a^{-1}b^{-1}c\ge 1$
if and only if $ab=c$. It was shown in \cite{abls} that $ab=c$ is a necessary condition
for separability of $\varrho_{a,b,c}$. Theorem \ref{kye-main-dual-norm} shows that this is
also sufficient for separability, without decomposing into the sum of product states.

In this section, we compute the dual norm in terms of entries in various
cases. We first deal with the vectors with real entries. To do this,
we consider the convex subset of the unit ball where the dual norm
is taken. More precisely, we define
$$
V_c=\{z\in\mathbb C^4: \|z\|_\xx\le 1,\ \re \lan c,z\ran =
\|c\|_\xx^\prime\},
$$
for a given $c\in\mathbb C^4$. The set $V_c$ is convex and nonempty,
and $z\in V_c$ implies $\|z\|_\xx=1$. Suppose that $c_i$ is a real
number for each $i=1,2,3,4$. Then we have $\re\lan c,z\ran =\re\lan
c,\bar z\ran$. Take any $z\in V_c$. Since $\|z\|_\xx=\|\bar
z\|_\xx$, we have $\bar z\in V_c$ and so, we have $\re z=\frac 12(z+\bar z)\in V_c$.
If $c\in\mathbb R^4$ then we have shown that there exists $z\in\mathbb R^4$ such that
$\re\lan c,z\ran=\|c\|_\xx^\prime$, and so this gives us
the formula
$$
\|c\|^\prime_\xx =\max\left\{ \frac{c_1z_1+c_2z_2+c_3z_3+c_4z_4}
{\max_\theta \left(|z_1e^{{\rm i}\theta} + z_4|+|z_2 e^{{\rm
i}\theta} + z_3|\right) }: z_1,z_2,z_3,z_4\in\mathbb R\right\},
$$
which has been already calculated in Section 5 of
\cite{han_kye_GHZ}. In order to explain this, we define the real
numbers
\begin{equation}\label{lambdas}
\begin{aligned}
\lambda_5 = 2(+c_1+c_2+c_3+c_4),\qquad
t_1&=c_1(-c_1^2+c_2^2+c_3^2+c_4^2)-2c_2c_3c_4,\\
\lambda_6 = 2(-c_1-c_2+c_3+c_4),\qquad
t_2&=c_2(+c_1^2-c_2^2+c_3^2+c_4^2)-2c_1c_3c_4,\\
\lambda_7 = 2(-c_1+c_2-c_3+c_4),\qquad
t_3&=c_3(+c_1^2+c_2^2-c_3^2+c_4^2)-2c_1c_2c_4,\\
\lambda_8 = 2(-c_1+c_2+c_3-c_4),\qquad
t_4&=c_4(+c_1^2+c_2^2+c_3^2-c_4^2)-2c_1c_2c_3,
\end{aligned}
\end{equation}
determined by $c\in\mathbb R^4$, and consider the following three cases:
\begin{enumerate}
\item[(A)]
$\lambda_5\lambda_6\lambda_7\lambda_8 \le 0$,
\item[(B)]
$\lambda_5\lambda_6\lambda_7\lambda_8 > 0$ and
$\left[t_1t_4\lambda_6\lambda_7 \ge 0\ {\rm or}\
t_2t_3\lambda_5\lambda_8 \le 0\right]$,
\item[(C)]
$\lambda_5\lambda_6\lambda_7\lambda_8 > 0$,
$t_1t_4\lambda_6\lambda_7 < 0$ and $t_2t_3\lambda_5\lambda_8 > 0$.
\end{enumerate}
The discussion in Section 5 of \cite{han_kye_GHZ} can be summarized
as follows:

\begin{proposition}\label{dual-norm-real-entry}
If $c_i$ is a real number for each $i=1,2,3,4$, then we have the
following:
\begin{enumerate}
\item[(i)]
in case of {\rm (A)} or {\rm (B)}, we have
$\|c\|_\xx^\prime=\|c\|_\infty$;
\item[(ii)]
in case of {\rm (C)}, we have $\|c\|_\xx^\prime=\frac 1{8}\Lambda(\lambda_5,\lambda_6,\lambda_7,\lambda_8)$,
where $\Lambda$ is in \eqref{eq:lambda}.
\end{enumerate}
\end{proposition}

Applying Proposition \ref{dual-norm} (iii), we can compute $\|c\|_\xx^\prime$
when at least one entry of $c$ is zero except for three entries, say $c_i,c_j,c_k$. In this case, it
is easily seen that the condition (A) for $\{|c_i|,|c_j|,|c_k|,0\}$ is satisfied if and only if
the three numbers $|c_i|, |c_j|$ and $|c_k|$ do not make a triangle. Indeed, we have
$$
\lambda_5\lambda_6\lambda_7\lambda_8=2^4(|c_3|^2-(|c_1|+|c_2|)^2)((|c_2|-|c_1|)^2-|c_3|^2).
$$
when $c_4=0$. Furthermore, the condition (B) for $\{|c_i|,|c_j|,|c_k|,0\}$ holds if and only if $|c_i|, |c_j|, |c_k|$ make an obtuse or right triangle, and the condition (C) for $\{|c_i|,|c_j|,|c_k|,0\}$ holds if and only if they make an acute triangle. Therefore, we have the
following:

\begin{theorem}\label{kye-three-entries-tri}
Suppose that at least one entry of $c\in\mathbb C^4$ is zero except for three entries
$c_i,c_j,c_k$. Then we have the following:
\begin{enumerate}
\item[(i)]
if the three numbers $|c_i|, |c_j|$ and $|c_k|$ do not make a triangle or make an obtuse or right triangle, then we have
$\|c\|_\xx^\prime=\|c\|_\infty$;
\item[(ii)]
if the three numbers $|c_i|, |c_j|$ and $|c_k|$ make an acute triangle, then we have
$$
\|c\|_\xx^\prime=
\frac{2|c_ic_jc_k|}{\sqrt{(|c_i|+|c_j|+|c_k|)(-|c_i|+|c_j|+|c_k|)(|c_i|-|c_j|+|c_k|)(|c_i|+|c_j|-|c_k|)}}.
$$
\end{enumerate}
\end{theorem}

\begin{corollary}\label{lower-esti,,,-dual}
If at least two entries of $c\in\mathbb C^4$ are zero, then $\|c\|_\xx^\prime=\|c\|_\infty$.
\end{corollary}

Now, we look for the formula of the dual norm
when the four entries are partitioned into two groups with two entries with the common magnitudes.
This is, of course, the counterpart of Theorem \ref{kye_two-partition}. Actually,
the following lemma makes it possible to use Theorem \ref{kye_two-partition} for our purpose.

\begin{lemma}\label{1221}
For $c=(c_1,c_2,c_2,c_1) \in \mathbb C^4$, we have
$$
\|c\|_{\sf X}' =
\sup \{ |\lan c,z\ran| : z=(z_1,z_2,z_2,z_1),~ \|z\|_{\sf X} \le 1 \}
$$
\end{lemma}

\begin{proof}
For a given $z\in\mathbb C^4$, put $w=\frac 12(z_1+z_4, z_2+z_3,z_2+z_3,z_1+z_4)\in\mathbb C^4$. Then we see that
$w_1=w_4$, $w_2=w_3$ and $\langle c,z\rangle=\langle c,w\rangle$. Therefore, the inequality
$$
\|w\|_\xx \le \frac 12\left(
\|(z_1,z_2,z_3,z_4)\|_\xx+\|(z_4,z_3,z_2,z_1)\|_\xx\right)=\|z\|_\xx
$$
gives the required result. Here, the last identity follows from Proposition \ref{kyeperm-char}.
\end{proof}

\begin{theorem}\label{criterion}
Let $c\in\mathbb C^4$ with $\phi_c\neq 0$. If there exists a partition
$\{i_1,i_2\}\cup\{i_3,i_4\}=\{1,2,3,4\}$ such that
$|c_{i_1}|=|c_{i_2}|=r$ and $|c_{i_3}|=|c_{i_4}|=s$ with $r,s\neq 0$,
then we have
$$
\|c\|_\xx^\prime
=\sqrt
{r^2t_0^2 + 2rst_0|\sin(\phi_c/2)|+s^2 \over t_0^2+1},
$$
where
$\displaystyle{t_0={r^2-s^2+\sqrt{(r^2-s^2)^2+(2rs\sin(\phi_c/2))^2} \over 2rs|\sin(\phi_c/2)|}}$.
If $\phi_c=0$ then $\|c\|_\xx^\prime=\max\{r,s\}$.
\end{theorem}

\begin{proof}
By the exactly same argument as in the proof of Theorem \ref{kye_two-partition},
we may assume that $|c_1|=|c_4|=r$ and $|c_2|=|c_3|=s$. Furthermore, we have
$$
\|c\|_\xx^\prime
= \|(re^{{\rm i}{\phi_c}/2},s,s,re^{{\rm i}{\phi_c}/2})\|_{\sf X}'
$$
by Proposition \ref{dual-norm} (i), so we may assume $c=(c_1,c_2,c_2,c_1)$ without loss of generality.
Let $z=(z_1,z_2,z_2,z_1)$ with $z_1=te^{{\rm i}\tau_1}$ and $z_2=e^{{\rm i}\tau_2}$. Then we have
$$
\|(z_1,z_2,z_2,z_1)\|_{\sf X}
=2\sqrt{t^2+1+2t|\cos(\tau_1-\tau_2)|}
$$
by Theorem \ref{kye_two-partition}. We also have
$$
|\lan c, z\ran|=2|c_1z_1+c_2z_2|=2|rte^{{\rm i}(\theta_1-\theta_2+\tau_1-\tau_2)}+s|.
$$
We put $\theta:=\theta_1-\theta_2$ and $\tau:=\tau_1-\tau_2$. Then $\phi_c=2\theta$ and $\phi_z=2\tau$,
and so we may assume that $\pi/2 \le \theta \le 3\pi/2$.
We define $f : [-\pi/2,3\pi/2] \times [0,\infty) \to \mathbb R$ by
\begin{equation}\label{dual-f-function}
\begin{aligned}
f(\tau,t)
=\left( {|\lan c, z\ran| \over \|z\|_{\sf X}} \right)^2
={|rte^{{\rm i}(\theta+\tau)} + s|^2 \over t^2 +1+2t|\cos \tau|}
={r^2t^2 + 2rst\cos(\theta+\tau)+s^2 \over t^2+1+2t|\cos\tau|}.
\end{aligned}
\end{equation}
It remains to find the maximum of $f$ on the domain $[-\pi/2,3\pi/2] \times [0,\infty)$.
Note that $f$ is not differentiable on the ray $\tau=\pi/2$.

We first consider on the open domain $(-\pi/2,\pi/2) \times (0,\infty)$, where we have
$$
\begin{aligned}
{\partial f \over \partial \tau}
&=
2t{(r^2\sin \tau -rs\sin(\theta+\tau))t^2 -(2rs\sin \theta) t +(s^2\sin \tau -rs\sin(\theta+\tau)) \over (t^2+1+2t\cos \tau)^2},
\\
{\partial f \over \partial t}
&=
2{(r^2\cos \tau -rs\cos(\theta+\tau))t^2+(r^2-s^2)t+(rs\cos(\theta+\tau)-s^2\cos\tau)
\over (t^2+1+2t\cos \tau)^2}.
\end{aligned}
$$
Let $A$ and $B$ be the numerators of ${1 \over 2t}{\partial f \over \partial \tau}$ and ${1 \over 2}{\partial f \over \partial t}$, respectively. From
$$
A^2+B^2=(r^2+s^2-2rs\cos \theta)(t^2+2t\cos \tau +1)(r^2t^2+2rst\cos(\theta+\tau)+s^2),
$$
we obtain a unique critical point
$$
\tau=\pi-\theta \in (-\pi/2,\pi/2) \quad \text{and} \quad t=s/r\in (0,\infty),
$$
where the function $f$ vanishes. When $\pi/2 < \tau < 3\pi/2$, we substitute $\tau'=\tau-\pi \in (-\pi/2,\pi/2)$ to get
$$
f(\tau,t)={r^2t^2 + 2rst\cos(\theta+\tau)+s^2 \over t^2+1-2t\cos\tau}
={r^2t^2 + 2rst\cos(\theta+\pi+\tau')+s^2 \over t^2+1+2t\cos\tau'}.
$$
Since $\pi - (\theta+\pi)=-\theta$ does not lie in $(-\pi/2,\pi/2)$,
the function $f$ does not have a critical point on the open domain $(\pi/2,3\pi/2)\times(0,\infty)$.

Next, we investigate the function $f$ on the boundaries $\tau=-\pi/2$ and $t=0$, and on the ray $\tau=\pi/2$.
Since
$$
f\left(\pm {\pi \over 2},t\right)={r^2t^2 \mp 2rst\sin\theta+s^2 \over t^2+1} \qquad \text{and} \qquad
f(\tau,0)=s^2=f\left(\pm {\pi \over 2},0\right),
$$
the maximum of $f$ occurs on the ray $\tau=\pi/2$ when $\pi \le \theta \le 3\pi/2$,
and on the ray $\tau=-\pi/2$ when $\pi/2 \le \theta \le \pi$.
We substitute $\tau=\pm \pi/2$ in the above ${\partial f \over \partial t}$ to get
$$
{\partial f \over \partial t}\left(\pm{\pi \over 2},t\right)=
2{\pm(rs\sin\theta) t^2 + (r^2-s^2)t \mp rs\sin\theta
\over (t^2+1)^2}.
$$
Solving the equation, we have
$$
t_0=
\begin{cases}
\displaystyle{{s^2-r^2-\sqrt{(s^2-r^2)^2+(2rs\sin\theta)^2} \over 2rs\sin\theta}} & \pi < \theta \le 3\pi/2, \\
\displaystyle{{s^2-r^2-\sqrt{(s^2-r^2)^2+(2rs\sin\theta)^2} \over
-2rs\sin\theta}} & \pi/2 \le \theta < \pi, \\
0 & \theta=\pi.
\end{cases}
$$
This completes the proof when $\phi_c\neq 0$, because $\theta=\phi_c/2$.
For the last case of $\phi_c=0$, the function $f$ has the maximum $s^2$ when $t=0$ or equivalently when $z_1=0$.
Note that our assumption $z_1=te^{{\rm i}\tau_1}, z_2=e^{{\rm i}\tau_2}$ does not cover the case $z_2=0$.
In this case, we obtain the maximum $r^2$ by the symmetry.
Alternatively, the case $\phi_c=0$ also comes out from Proposition \ref{dual-norm-real-entry}.
\end{proof}

Suppose that $|c_i|=r$ for each $i=1,2,3,4$. Then we have $t_0=1$ in Theorem \ref{criterion} and we have $\|c\|_\xx^\prime=
r \sqrt{1+|\sin(\phi_c/2)|}$.
Compare with (\ref{X-norm-share-mag}).
Therefore, we see that an {\sf X}-state $\varrho=X(a,b,c)$
is separable if and only if
\begin{equation}\label{share-mag-criterion}
\Delta_\varrho\ge r \sqrt{1+|\sin(\phi_\varrho/2)|},
\end{equation}
when all the anti-diagonals have the magnitude $r$.
This recover the result in \cite{han_kye_phase} without a decomposition.
We also consider the state
$$
\varrho_{\theta,r,s}=X({\bf 1},{\bf 1}, (e^{{\rm i}\theta}r,r,s,s)).
$$
Theorem \ref{kye-main-dual-norm} tells us that
$\varrho_{0,r,s}$ is separable if and only if $\max\{r,s\}\le 1$ and $\phi_{\pi,r,s}$ is separable if and only if
$r^2+s^2\le 1$, because
$\|(-r,r,s,s)\|_\xx^\prime=\sqrt{r^2+s^2}$ by Theorem \ref{criterion} or Proposition \ref{dual-norm-real-entry}.
Note that Theorem \ref{criterion} gives a interpolation between these two cases, $\theta=0$ and $\theta=\pi$. See Figure 2.

\begin{figure}[t]
\centering
\centerline{\includegraphics[width=2in]{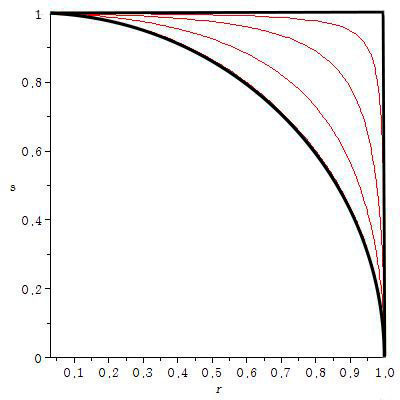}}
\caption{
Two thick curves, circle and rectangle, represent the boundary of the separability region $(r,s)$ of
the state $\varrho_{\theta,r,s}$
for $\theta=\pi$ and $\theta=0$, respectively. The other curves represent the separability regions
for $\theta=\pi/2,\pi/4,\pi/10$.
 \label{overflow}}
\end{figure}

It seems to be very difficult to compute the dual norm $\|c\|_\xx^\prime$ in general cases.
Because $\|c\|_\infty$ is a lower bound for the dual norm by (\ref{dual-basic-bound}),
it is natural to look for condition under which
the equality $\|c\|_\xx^\prime= \|c\|_\infty$ holds.

\begin{proposition}\label{dual-norm-bound}
For each $c\in\mathbb C^4$, we have
$\|c\|_\infty\le\|c\|_\xx^\prime\le \|(-|c_1|,|c_2|,|c_3|,|c_4|)\|_\xx^\prime$.
\end{proposition}

\begin{proof}
It suffices to prove the second inequality.
Let $\tilde c = (-|c_1|,|c_2|,|c_3|,|c_4|)$ and $\tilde z = (-|z_1|,|z_2|,|z_3|,|z_4|)$.
We have
$$
\frac{{\rm Re}\, \lan c,z\ran}{\|z\|_\xx} \le\frac{\lan \tilde c,
\tilde z\ran}{\|z\|_\xx} \le\frac{\lan \tilde c,
\tilde z\ran}{\|\tilde z\|_\xx} \le \|\tilde c\|_\xx^\prime
$$
for every $z\in\mathbb C^4$. Here, the second inequality holds since
we have $\|\tilde z\|_\xx\le \|z\|_\xx$ by Proposition
\ref{beta_decrease}.
\end{proof}

We can obtain the closed formula of $ \|(-|c_1|,|c_2|,|c_3|,|c_4|)\|_\xx^\prime$ using Proposition
\ref{dual-norm-real-entry}. In particular, if $(-|c_1|,|c_2|,|c_3|,|c_4|)$ satisfies the condition (A) or
(B), then $\|c\|_\xx^\prime=\|c\|_\infty$ by Proposition \ref{dual-norm-bound}. In order to find this
condition, we consider the case of $(-c_1,c_2,c_3,c_4)$ with
positive $c_i$'s. In this case, the numbers $\l_i$ and $t_j$ in (\ref{lambdas}) determined by $c\in\mathbb R^4$
become $\lambda^-_i$
and $t^-_j$ as follows.
\begin{equation}
\begin{aligned}\label{-lambdas}
\lambda^-_5 = 2(-c_1+c_2+c_3+c_4),\qquad
-t_1^-&=c_1(-c_1^2+c_2^2+c_3^2+c_4^2)+2c_2c_3c_4,\\
\lambda^-_6 = 2(+c_1-c_2+c_3+c_4),\qquad
\phantom{-}t_2^-&=c_2(+c_1^2-c_2^2+c_3^2+c_4^2)+2c_1c_3c_4,\\
\lambda^-_7 = 2(+c_1+c_2-c_3+c_4),\qquad
\phantom{-}t_3^-&=c_3(+c_1^2+c_2^2-c_3^2+c_4^2)+2c_1c_2c_4,\\
\lambda^-_8 = 2(+c_1+c_2+c_3-c_4),\qquad
\phantom{-}t_4^-&=c_4(+c_1^2+c_2^2+c_3^2-c_4^2)+2c_1c_2c_3.
\end{aligned}
\end{equation}
We also consider the form
$$
Q[a,b,c](x)= -x^3+(a^2+b^2+c^2)x+2abc
$$
with positive variables $a,b,c$ and $x$. This is invariant under
all the permutations of variables $a,b$ and $c$. We have
\begin{equation}\label{kye-def_u_i}
-t_1^-=Q[c_2,c_3,c_4](c_1)=:u_1,\qquad
t_i^-=Q[c_j,c_k,c_\ell](c_i)=:u_i\ (i=2,3,4),
\end{equation}
where $\{i,j,k,\ell\}=\{1,2,3,4\}$.
By elementary calculus, we see that the function $Q[a,b,c]$ has
exactly one root $\alpha[a,b,c]$ in the interval $[0,\infty)$, and
$$
Q[a,b,c](x)<0\quad \Longleftrightarrow\quad x>\alpha[a,b,c].
$$
It is easy to check that
$$
Q[a,b,c](a+b+c)<0,\quad Q[a,b,c](a)>0,\quad Q[a,b,c](b)>0,\quad Q[a,b,c](c)>0,
$$
which implies
$$
a+b+c>\alpha[a,b,c],\quad a<\alpha[a,b,c],\quad b<\alpha[a,b,c],\quad c<\alpha[a,b,c].
$$
If $u_i=Q[c_j,c_k,c_\ell](c_i)\le 0$ ($i=1,2,3,4$), then we have
$c_i \ge \alpha[c_j,c_k,c_\ell]>c_j,c_k,c_\ell$. Therefore, two of
$u_i$'s cannot be non-positive, and so we conclude that all of $u_i$'s
are positive or exactly one of them is non-positive. It is also clear
that $u_i=Q[c_j,c_k,c_\ell](c_i)\le 0$ holds if and only if
$c_i\ge\alpha[c_j,c_k,c_\ell]>c_j,c_k,c_\ell$.

\begin{proposition}\label{quadrangle}
Let $c_j > 0$ for $j=1,2,3,4$.
Suppose that $\lambda^-_i$ and $t^-_i$ are defined by {\rm
(\ref{-lambdas})}. Then we have the following:
\begin{enumerate}
\item[(i)]
$\prod_{i=5}^8\lambda^-_i > 0$ if and only if $c_i$'s make a
quadrangle if and only $\lambda^-_i>0$ for each $i=5,6,7,8$;
\item[(ii)]
$(-c_1,c_2,c_3,c_4)$ satisfies the condition {\rm
(A)} or {\rm (B)} if and only if exactly one of $u_i$'s in {\rm (\ref{kye-def_u_i})} is non-positive.
\end{enumerate}
\end{proposition}

\begin{proof}
Let $\{i,j,k,\ell\}=\{1,2,3,4\}$.
The sum of any two $\lambda_i^-$ is positive.
Statement (i) is easy to check. For the `only if' part of (ii), we
first consider the condition (A), namely $\prod_{i=5}^8\lambda^-_i\le 0$.
Then we have $\lambda^-_i \le 0$ for some $i$, that is, $c_i\ge c_j+c_k+c_\ell$.
Since $c_i\ge c_j+c_k+c_\ell > \alpha[c_j,c_k,c_\ell]$, we have $u_i=Q[c_j,c_k,c_\ell](c_i)<0$.
Suppose that $\prod_{i=5}^8\lambda^-_i> 0$. Then the condition (B)
holds if and only if $(-t^-_1)t^-_4 \le 0$ or $t^-_2t^-_3 \le 0$ if and only
if one of $u_i$ is non-positive.

For the `if' part of (ii), suppose that $u_i\le 0$, $u_j>0$, $u_k>0$
and $u_\ell>0$. Then we have $c_j, c_k, c_\ell<\alpha[c_j, c_k,
c_\ell]\le c_i$. If $c_i\ge c_j+c_k+c_\ell$, then $c_i$'s do not make
a quadrangle. If $c_i< c_j+c_k+c_\ell$, then they make a quadrangle
and satisfy the condition (B). This completes the proof.
\end{proof}

We combine Proposition \ref{dual-norm-bound} and the Proposition \ref{quadrangle} to get a sufficient condition on $c\in\mathbb C^4$ under which the equality
$\|c\|_\xx^\prime=\|c\|_\infty$ holds.

\begin{theorem}\label{kye-dual-infty}
Suppose that $|c_k|>0$ for $k=1,2,3,4$. Define the real numbers
$v_k$ by
$$
v_k=|c_k|\left(-2|c_k|^2+\sum_{i=1}^4|c_i|^2\right)+2\prod_{i \ne k} |c_i|,
$$
for $k=1,2,3,4$. Then we have the following:
\begin{enumerate}
\item[(i)]
either all of $v_k$'s are positive or exactly one of them is non-positive;
\item[(ii)]
if one of $v_k$'s is non-positive, then $\|c\|_\xx^\prime=\|c\|_\infty$;
\item[(iii)]
if $|c_k|$'s does not make a quadrangle, then one of $v_k$'s is non-positive.
\end{enumerate}
\end{theorem}


\section{Estimates of the norms}
\label{sec:norm-estimate}

In the previous section, we have computed the dual norm in some cases
in order to provide necessary and sufficient criteria
for separability which is ready to apply with routine calculations. The formulae are too complicated
to understand in some cases. For example, see the formula given in Theorem \ref{criterion}.
Furthermore, it seems to be hopeless to get exact formula for the dual norm in general cases.
In this sense, it is worthwhile to estimate the dual norm. Especially,
lower bounds for the dual norm will give necessary criteria for separability.
In this section, we try to estimate the norms $\|\ \|_\xx$ and $\|\ \|_\xx^\prime$.

Theorem \ref{criterion} itself is quite useful to estimate the dual norm.
For a given $c=(r_1e^{{\rm i}\phi},r_2,r_3,r_4)$, we consider
$$
c^\prime=\left({r_1+r_4 \over 2}e^{{\rm i}\phi},{r_2+r_3 \over 2},
    {r_2+r_3 \over 2},{r_1+r_4 \over 2}\right).
$$
Then we have
$$
\begin{aligned}
\|c\|_\xx^\prime
&= \|(r_1e^{{\rm i}\phi/2},r_2,r_3,r_4e^{{\rm i}\phi/2})\|_{\sf X}' \\
&= {1 \over 2}\left(\|(r_1e^{{\rm i}\phi/2},r_2,r_3,r_4e^{{\rm i}\phi/2})\|_{\sf X}'
  + (r_4e^{{\rm i}\phi/2},r_3,r_2,r_1e^{{\rm i}\phi/2})\|_{\sf X}'\right) \\
&\ge \|\textstyle{({r_1+r_4 \over 2}e^{{\rm i}\phi/2},{r_2+r_3 \over 2},
    {r_2+r_3 \over 2},{r_1+r_4 \over 2}e^{{\rm i}\phi/2})}\|_{\sf X}'=\|c^\prime\|_\xx^\prime,
\end{aligned}
$$
by Proposition \ref{dual-norm} (i) and Proposition \ref{kyeperm-char}. By Theorem \ref{kye-main-dual-norm}, we have the following:

\begin{proposition}\label{2+2->arbi}
If an {\sf X}-state $X(a,b,c)$ is separable, then $X(a,b,c^\prime)$ is also separable.
\end{proposition}

Note that $c^\prime$ satisfies the assumption of Theorem \ref{criterion}. The {\sf X}-part of a separable three qubit state is again separable \cite{han_kye_GHZ}. Thus, we may get a necessary
criterion for a three qubit state whose {\sf X}-part is given by $X(a,b,c)$ for arbitrary $c\in\mathbb C^4$.
Applying Theorem \ref{criterion} directly, we have the necessary criterion by the inequality
$$
\|c\|_\xx^\prime \ge
\sqrt{{m_1^2t_0^2 + 2m_1m_2t_0|\sin(\phi_c/2)|+m_2^2 \over t_0^2+1}},
$$
where
$m_1={r_1+r_4 \over 2}$, $m_2={r_2+r_3 \over 2}$ and
$t_0={m_1^2-m_2^2+\sqrt{(m_1^2-m_2^2)^2+(2m_1m_2\sin(\phi_c/2))^2} \over 2m_1m_2|\sin(\phi_c/2)|}$.
By the exactly same argument as in the proof of Theorem \ref{kye_two-partition}, this inequality also holds
when $m_1={r_{i_1}+r_{i_2} \over 2}$, $m_2={r_{i_3}+r_{i_4} \over 2}$.
If we put a suitable $(\tau,t)$ in the function $f$ in (\ref{dual-f-function}), then we may get
interesting separability criteria. We exhibit two applications in this direction.

\begin{proposition}\label{gggjghhl}
Suppose that $\varrho$ is a three qubit state with the {\sf X}-part $X(a,b,c)$.
If $\varrho$ is separable then the inequality
$$
\Delta_\varrho \ge {\sqrt{2}m_1m_2 \over \sqrt{m_1^2+m_2^2}} \sqrt{1+|\sin(\phi_c/2)|}
$$
holds for
$m_1={|c_{i_1}|+|c_{i_2}| \over 2}$ and $m_2={|c_{i_3}|+|c_{i_4}| \over 2}$.
\end{proposition}

\begin{proof}
We first consider the case when $r_1=r_4=r$ and $r_2=r_3=s$, and get a necessary condition
$$
\Delta_\varrho^2 \ge
f\left(\pm {\pi \over 2},{s \over r}\right) = {2r^2s^2\left(1\mp \sin (\phi_c/2)\right) \over r^2+s^2}.
$$
By Proposition \ref{2+2->arbi}, the assertion holds when
$m_1={|c_{1}|+|c_{4}| \over 2}$ and $m_2={|c_{2}|+|c_{3}| \over 2}$. The other cases follow by the same argument
as in the proof of Theorem \ref{kye_two-partition}.
\end{proof}

If the anti-diagonals share a common magnitude $r$, then
Proposition \ref{gggjghhl} enables us to recover the `only if' part of (\ref{share-mag-criterion}).
The next one has an interesting geometric interpretation.

\begin{proposition}\label{,hluighkfgdgh}
Let $\varrho$ be a three qubit state with the {\sf X}-part $X(a,b,c)$.
Suppose that the triangle with two sides $|c_{i_1}|+|c_{i_2}|$, $|c_{i_3}|+|c_{i_4}|$ and the internal angle $|\phi_c/2|$ is acute.
If $\varrho$ is separable, then we have
$$
\Delta_\varrho\ge R,
$$
where $R$ is the radius of the circumscribed circle.
\end{proposition}

\begin{proof}
We first consider the case when $r_1=r_4=r$ and $r_2=r_3=s$.
Since the triangle is acute, we have $2r>2s\cos (\phi_c/2)$ and $2s>2r\cos (\phi_c/2)$.
We take
$$
\theta=\phi_c/2+\pi \in (\pi/2,3\pi/2), \quad \tau=-\theta, \quad
t={r+s\cos \theta \over s+r \cos \theta}={r-s\cos (\phi_c/2) \over s-r \cos (\phi_c/2)}>0,
$$
and put in $f$ in (\ref{dual-f-function}). We get the condition
$$
\Delta_\varrho^2 \ge
f(\tau,t) =
{r^2+s^2-2rs\cos(\phi_c/2) \over \sin^2(\phi_c/2)}.
$$
By the laws of sines and cosines, the right hand side is equal to $R^2$.
For general cases, we use the same argument as in the proof of Proposition \ref{gggjghhl}.
\end{proof}

Now, we turn our attention to estimate the norm $\|\ \|_\xx$. First of all, the following inequality
\begin{equation}\label{esti-basic}
\|z\|_\infty \le \|z\|_\xx \le \|z\|_1,\qquad z\in\mathbb C^4
\end{equation}
is clear by the definition of the norm. We have already seen that
the upper bound $\|z\|_1$ is sharp in several cases by Proposition
\ref{X-norm}. In order to get other lower bounds, we consider
the natural projection map $\pi_S$ from $\mathbb C^4$ onto
$|S|$-dimensional space by taking entries whose indices belong to
$S$, where $|S|$ denotes the cardinality of $S$. For example,
$\pi_{\{1,3\}}$ maps $(z_1,z_2,z_3,z_4)\in\mathbb C^4$ to $(z_1,z_3)\in\mathbb C^2$.

\begin{proposition}\label{lower-esti,,,}
If $|S|\le 2$ then
the map $\pi_S:(\mathbb C^4, \|\ \|_\xx)\to (\mathbb C^{\abs{S}}, \|\ \|_1)$ is norm decreasing.
That is, we have the inequality
$\|z\|_\xx\ge \|\pi_S(z)\|_1$ for every $z\in\mathbb C^4$.
\end{proposition}

\begin{proof}
If $|S|=1$ then we have $\|z\|_\xx\ge \|z\|_{\infty}\ge\|\pi_S(z)\|_1$. Let $|S|=2$.
We write $|z_i|=s_i$.
Evidently $\| z\|_\xx\ge \max\{s_1+s_4,s_2+s_3\}$, and so the inequality holds when
$S=\{1,4\}$ or $S=\{2,3\}$. For the remaining cases, we note
$$
\| z\|_\xx=\max_\sigma \bigg(
\sqrt{s_1^2+s_4^2+2s_1s_4\cos(\phi_z+\s)} +
\sqrt{s_2^2+s_3^2+2s_2s_3\cos\s} \bigg),
$$
by \eqref{eq:zx}.
If $\phi_z \in [0,\p/2] \cup [3\p/2,2\p]$ then the inequalities hold
by choosing $\s=0$. If $\phi_z \in [\p/2,\p]$ then the inequalities also
hold by choosing $\s=-\p/2$. Finally, we have the inequalities by choosing $\s=\p/2$
when $\phi_z \in [\p,3\p/2]$.
\end{proof}

Proposition \ref{lower-esti,,,} does not hold when $|S|=3$.
For example, we have $\|(1,1,1,-1)\|_\xx=2\sqrt 2$ by
(\ref{X-norm-phase_pi}) or (\ref{X-norm-share-mag}), and $\|\pi_{\{1,2,3\}}(1,1,1,-1)\|_1=3$.
We also have $\|(1,1,1,0)\|_\xx=3$ by Proposition \ref{X-norm} (ii). So, it
should be noted that $\|(z_1,z_2,z_3,0)\|_\xx$ may exceed
$\|(z_1,z_2,z_3,z_4)\|_\xx$. Proposition \ref{lower-esti,,,} suggests to introduce
\begin{equation}\label{special-norm}
\|z\|_\square
=\max\{\|\pi_S(z)\|_1:S\subset\{1,2,3,4\},\ |S|=2\}.
\end{equation}
It is easily checked that $\|\ \|_\square$ is a norm on the vector space $\mathbb C^4$. In fact, we have
$\|z\|_\square=\max_{j,k=1,2,3,4,~j\ne k} \{\abs{z_j}+\abs{z_k} \}$.
Then Proposition \ref{lower-esti,,,} can be written by the inequality
$\|(z_1,z_2,z_3,z_4)\|_\xx\ge \|z\|_\square$.

Using Proposition \ref{lower-esti,,,}, we can prove Corollary \ref{lower-esti,,,-dual} by duality,
To see this, we note that the dual map $\iota_S:\mathbb C^2\to \mathbb C^4$ of $\pi_S$
is the natural embedding, for example, $\iota_{\{1,4\}}(\gamma_1,\gamma_2)=(\gamma_1,0,0,\gamma_2)$.
It is also norm-decreasing from $(\mathbb C^2,\|\ \|_\infty)$ to $(\mathbb C^4,\|\ \|_\xx^\prime)$,
by Proposition \ref{lower-esti,,,}.
If at least two entries of $c$ is zero
then we can take $S\subset \{1,2,3,4\}$ and $\gamma\in\mathbb C^2$ such that $\iota_S(\gamma)=c$.
Then we have
$\|c\|_\infty=\|\gamma\|_\infty\ge \|\iota_S(\gamma)\|_\xx^\prime=\|c\|_\xx^\prime$.
The other inequality
follows from
\begin{equation}\label{dual-basic-bound}
\|c\|_\infty \le \|c\|_\xx^\prime \le \|c\|_1,
\end{equation}
which is the dual statement of (\ref{esti-basic}).

Proposition \ref{beta_decrease} shows that the function $\beta$ takes the minimum at $\phi=\pi$, and so
it is natural to seek for lower bounds of two expressions in (\ref{X-norm-phase_pi})
in order to find lower bounds for $\|\ \|_\xx$.
For the second expression, we consider the function
$f(z)=\displaystyle{\frac{-z_1+z_2+z_3+z_4}{\Lambda(-z_1,z_2,z_3,z_4)}}$ with $z_i \in \mathbb R$ and $z_1z_2z_3z_4<0$,
which appears in the Appendix of \cite{han_kye_GHZ} with $c=(-1,1,1,1)$. It was shown in \cite{han_kye_GHZ}
that $f$ has the maximum $\sqrt 2$ on the ray $\{\lambda(-1,1,1,1) : \lambda>0\}$.
This shows the inequality
\begin{equation}\label{ineq-gen-1}
\Lambda(a_1,a_2,a_3,a_4)\ge \frac 1{\sqrt 2}(a_1+a_2+a_3+a_4),\qquad a_i>0,
\end{equation}
where the equality holds if and only if $a_1=a_2=a_3=a_4$.
On the other hand, we have the inequality
\begin{equation}\label{ineq-gen-2}
\Lambda(a_1,a_2,a_3,a_4)\ge (a_1+a_2+a_3+a_4)-2a_i,\qquad a_i>0,
\end{equation}
from the identity
$$
\begin{aligned}
&(a_1a_2+a_3a_4)(a_1a_3+a_2a_4)(a_1a_4+a_2a_3)-a_1a_2a_3a_4(a_1+a_2+a_3+a_4-2a_i)^2\\
&\phantom{XXXXXXXXXXXXXXXXX}
=(a_1a_2a_3a_4)^2
\left(\frac 1{a_1}+\frac 1{a_2}+\frac 1{a_3}+\frac 1{a_4}-\frac 2{a_i}\right)^2.
\end{aligned}
$$
For the first expression in (\ref{X-norm-phase_pi}), we also note that
\begin{equation}\label{ineq-gen-3}
\frac 1{a_i}\ge\sum_{j\neq i}\frac 1{a_j}
\ \Longrightarrow\
a_1+a_2+a_3+a_4-2a_i\ge\frac 1{\sqrt 2}(a_1+a_2+a_3+a_4).
\end{equation}
Indeed, for $i=1$, the assumption part of (\ref{ineq-gen-3}) implies
$$
\begin{aligned}
a_2+a_3+a_4
&\ge a_1\left(\frac 1{a_2}+\frac 1{a_3}+\frac 1{a_4}\right)(a_2+a_3+a_4)\\
&=a_1\left(3+\frac{a_2}{a_3}+\frac{a_3}{a_2}+\frac{a_3}{a_4}+\frac{a_4}{a_3}+\frac{a_4}{a_2}+\frac{a_2}{a_4}\right)\\
&\ge 9a_1\ge(\sqrt 2+1)^2a_1,
\end{aligned}
$$
from which we get the conclusion part of (\ref{ineq-gen-3}).
We have $\|z\|_\xx\ge \|(-|z_1|, |z_2|, |z_3|, |z_4|)\|_\xx$
by Proposition \ref{beta_decrease}. Therefore,
the formula (\ref{X-norm-phase_pi}) together with
(\ref{ineq-gen-1}), (\ref{ineq-gen-2}) and (\ref{ineq-gen-3}) implies the following:

\begin{proposition}\label{lower-est-sqrt_2}
For every $z\in\mathbb C^4$, we have the following:
\begin{enumerate}
\item[(i)]
$\|z\|_1\ge \|z\|_\xx\ge \frac 1{\sqrt 2}\|z\|_1$;
\item[(ii)]
$\|z\|_\xx\ge \sum^4_{j=1} \abs{z_j} - 2\min_j \abs{z_j}$.
\end{enumerate}
\end{proposition}

The lower bound in
Proposition \ref{lower-est-sqrt_2} (ii) is stronger than that in Proposition \ref{lower-esti,,,}.
Further, the lower bounds in Proposition \ref{lower-esti,,,} and Proposition
\ref{lower-est-sqrt_2} (i) are independent. For example, we consider
$z=(p,q,q,q)$ with real numbers $p$ and $q$. When $p$ and $q$ share the same sign, we have $\|z\|_\xx=|p|+3|q|$ by Proposition \ref{X-norm} (iii). When their signs are different, we evaluate the norm
$\|z\|_\xx$ with the formula (\ref{X-norm-phase_pi}). Then we see
that $z\in \Omega^-_1$ if and only if $|q|\ge 3|p|$, and
$\Omega^-_2$, $\Omega^-_3$ and $\Omega^-_4$ are empty. Furthermore,
if $z\in\Omega^-_0$ then $\|z\|_\xx\le 1$ if and only if
$(|p|+|q|)^3\le |p|$. With this information, we may figure out the
region of $(p,q)$ so that $(p,q,q,q)$ belongs to the unit ball with
respect to the norm $\|\ \|_\xx$. See Figure 3.

\begin{figure}[t]
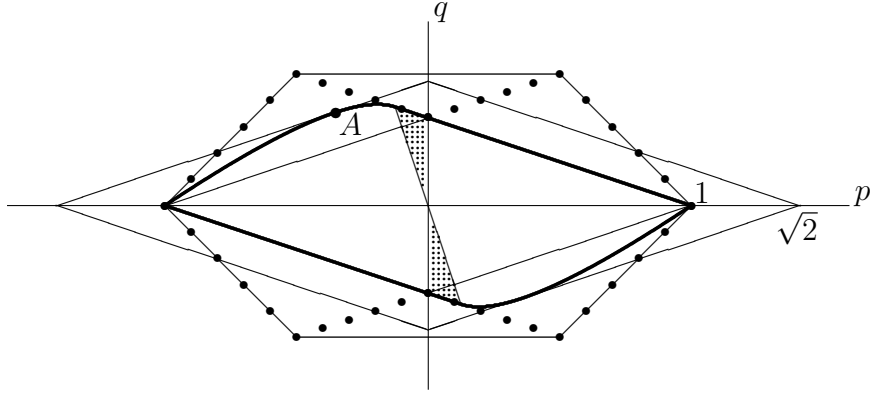

\centering
\input X-state_figure.tex
\end{picture}
\end{center}
\caption{ The region surrounded by thick lines and curves denotes $(p,q)$
so that $\|(p,q,q,q)\|_\xx\le 1$.
Proposition \ref{lower-est-sqrt_2} (i) tells us that the boundary is
sitting between two diamonds which represent the balls of radii $1$
and $\sqrt 2$ with respect to the $\ell^1$-norm, respectively. On
the first and third quadrants, the phase difference is $0$, and so
the boundary coincides with that of the $\ell^1$-unit ball by
Proposition \ref{X-norm} (iii). On the second and fourth quadrants,
the dotted and undotted regions are determined by $\Omega_1^-$ and
$\Omega_0^-$ in the formula (\ref{X-norm-phase_pi}), respectively.
The point $A\textstyle{\left(-\frac
{1}{2\sqrt 2},\frac 1{2\sqrt 2}\right)}$ is on the boundary because
$\|(-1,1,1,1)\|_\xx=2\sqrt 2$ by (\ref{X-norm-phase_pi}) or
(\ref{X-norm-share-mag}). The hexagon represents the unit ball with
respect to the norm $\|\ \|_\square$ defined in (\ref{special-norm}).
Finally, the concave polygon given  by dots represents the region determined by
Proposition \ref{lower-est-sqrt_2} (ii).
}
\end{figure}

Next, we consider the relation between the norm $\|c\|_\xx$ and its dual norm $\|c\|_\xx^\prime$.
This will gives us lower bounds for the dual norm $\|\ \|_\xx^\prime$ in terms of $\|\ \|_\xx$.

\begin{proposition}
For each $c\in\mathbb C^4$, we have
$\displaystyle{\|c\|_\xx^\prime
\ge\max_\phi\frac{\|(c_1,c_2,c_3,c_4 e^{{\rm i}\phi})\|_\xx}{2\sqrt 2\sqrt{1+|\cos (\phi/2)|}}}$.
\end{proposition}

\begin{proof}
We take
$z_{\theta,\phi}=
(e^{{\rm i}\theta}e^{{\rm i}\tau}, e^{{\rm i}\theta}e^{{\rm i}\psi}, e^{{\rm i}\tau}, e^{{\rm i}\phi}e^{{\rm i}\psi})\in\mathbb C^4$
with
$\tau=-{\rm arg} (c_1e^{{\rm i}\theta}+c_3)$ and $\psi=-{\rm arg}(c_2e^{{\rm i}\theta}+c_4e^{{\rm i}\phi})$,
for given $\theta$ and $\phi$.
By definition we have
$\displaystyle{\|c\|_\xx^\prime
\ge\max_{\theta,\phi}\frac{{\rm Re}\,\lan c,z_{\theta,\phi}\ran}{\|z_{\theta,\phi}\|_\xx}}$.
Note that
$$
\lan c, z_{\theta,\phi}\ran
=(c_1e^{{\rm i}\theta}+c_3)e^{{\rm i}\tau}+(c_2e^{{\rm i}\theta}+c_4e^{{\rm i}\phi})e^{{\rm i}\psi}
=|c_1e^{{\rm i}\theta}+c_3|+|c_2e^{{\rm i}\theta}+c_4e^{{\rm i}\phi}|,
$$
and so, it follows that
$\max_\theta \lan c, z_{\theta,\phi}\ran
=\|(c_1,c_2,\bar c_4e^{-{\rm i}\phi},\bar c_3)\|_\xx
=\|(c_1,c_2,c_3,c_4e^{{\rm i}\phi})\|_\xx
$
by Proposition \ref{fytygdgfnd} (ii). Because the phase difference of $z_{\theta,\phi}$ is given by $\phi$, we also have
$\|z_{\theta,\phi}\|_\xx=2\sqrt 2\sqrt{1+|\cos (\phi/2)|}$ by (\ref{X-norm-share-mag}).
\end{proof}

If we take $\phi=\pi$, then we have the following estimate:
$$
\displaystyle{\|c\|_\xx^\prime\ge \frac 1{2\sqrt 2}\|(c_1,c_2,c_3,-c_4)\|_\xx},
\qquad
c\in\mathbb C^4.
$$
We close this section to give two upper bounds for the dual norm $\|c\|_\xx^\prime$,
which are sharper than $\|c\|_\xx^\prime\le \|c\|_1$.
If $|c_{i_1}|\le|c_{i_2}|\le|c_{i_3}|\le|c_{i_4}|$, then we use Proposition \ref{lower-esti,,,} to obtain
$$
\begin{aligned}
\abs{\lan c,z\ran}
\le\sum_{i=1}^4|c_i||z_i|
&\le |c_{i_2}|(|z_{i_1}|+|z_{i_2}|)+|c_{i_4}|(|z_{i_3}|+|z_{i_4}|)\\
&\le |c_{i_2}|\|z\|_\xx +|c_{i_4}|\|z\|_\xx
=( |c_{i_2}| +|c_{i_4}|)\|z\|_\xx.
\end{aligned}
$$
So, we have an upper bound $\|c\|_\xx^\prime\le |c_{i_2}|+|c_{i_4}|$
whenever $|c_{i_1}|\le|c_{i_2}|\le|c_{i_3}|\le|c_{i_4}|$.
Finally, we have another upper bound $\|c\|_\xx^\prime\le \sqrt 2\|c\|_\infty$
by Proposition \ref{lower-est-sqrt_2} (i) and the duality.


\section{Discussion}

The duality between tensor products and linear maps plays a central role in the separability
problem through the bi-linear pairing
$\langle\varrho,\phi\rangle$, as it was found by Horodecki's \cite{horo-1}.
See also \cite{eom-kye}.
Here, the bi-linear pairing  can be described in terms of
Choi matrix $C_\phi$ of the map $\phi$ through $\langle\varrho,\phi\rangle={\text{\rm tr}}(\varrho C_\phi^\ttt)$.
This duality tells us that a state $\varrho$ is separable
if and only if $\langle\varrho,\phi\rangle \ge 0$ for every positive linear map $\phi$.
In other words, a state $\varrho$ is entangled if and only if there exists a positive linear map $\phi$
such that $\langle\varrho,\phi\rangle < 0$.

Therefore, this duality is very useful to
detect entanglement, and this has been formulated as the notion of entanglement witnesses
\cite{terhal}, which is nothing but the Choi matrix of a positive linear map.
On the other hand, we have to know all
positive linear maps, in order to
show that a given state is separable with this duality.
This was done for $2\otimes 2$ and $2\otimes 3$ cases
by the classical results of St{\o}rmer \cite{stormer} and Woronowicz \cite{woronowicz}, as they showed that
every positive linear map between $2\times 2$ and $3\times 3$ matrices is a decomposable
positive map, which has an exact form.

The above mentioned duality between bi-partite separable states and positive linear maps
has been extended \cite{kye_3qb_EW} to the duality between $n$-partite
separable states and positive multi-linear
maps with $(n-1)$ variables. In the three qubit case,
we know \cite{han_kye_tri} the condition for an {\sf X}-shaped self-adjoint matrices
under which they are Choi matrices of a positive bi-linear maps.
This enables us to show that the inequality (\ref{main-cri-dual-norm,,,})
is a sufficient condition for separability.
This is one of very few cases to find a sufficient condition without decomposition
into the sum of pure product states, as it was mentioned in Introduction.

Nevertheless, it is of an independent interest to look for decomposition.
For some separable three qubit {\sf X}-states, we already know decompositions.
This is the case for most GHZ-diagonal states \cite{guhne_pla_2011},
for those whose rank is less than or equal to six \cite{han_kye_phase},
for those whose anti-diagonal entries share a common magnitude \cite{han_kye_phase}.
We exhibit one more case. We have shown that the state $\varrho_{a,b,c}$ defined in (\ref{acin-exam})
is separable if and only if $ab=c$.
In order to decompose these states, we note that
\bea
\label{eq:psj}
&&
\ket{\ps_j}
\notag\\
&=&
u_{j0}(\ket{000}+\ket{111})
+
u_{j1}\ket{001}
+
u_{j2}\ket{010}
+
u_{j3}\ket{011}
+
u_{j4}\ket{100}
+
u_{j5}\ket{101}
+
u_{j6}\ket{110}
\notag\\
&=&
{1\over\sqrt7}
(\ket{0}+\omega^{4j}\ket{1}) \otimes
(\ket{0}+\omega^{2j}\ket{1}) \otimes
(\ket{0}+\omega^{j}\ket{1})
\notag\\
\eea
is a product vector for $j=0,1,\dots,6$, when
$u_{jk}={1\over\sqrt7}\omega^{jk}$ with the seventh root $\omega=e^{2\pi {\rm i} \over 7}$ of unity. That is, $[u_{jk}]$
is an order-seven discrete Fourier transform.
One can verify the identity
$$
\varrho_{1,1,1}=\sum^6_{j=0}\proj{\ps_j},
$$
which tells us that $\varrho_{1,1,1}$ can be decomposed into the sum of seven pure product states.
Because $\varrho_{1,1,1}$ has rank seven, this is an optimal decomposition, that is,
the number of pure product states is smallest among all the possible decompositions.
Note that two states $\varrho_{1,1,1}$ and $\varrho_{a,b,c}$
are equivalent under the local invertible operator
$P=\diag(1,c^{-1/2})\ox\diag(1,b^{1/2})\ox\diag(1,a^{1/2})$ whenever $ab=c$.
Therefore, we can also get an optimal decomposition of $\varrho_{a,b,c}$
into the sum of seven product states.

We recall that the length of a separable state is defined by the number
of pure product states appearing in an optimal decomposition. The length should be greater than or equal to
the rank of $\varrho$.
It was shown in \cite{chen_dj_2xd,chen_dj_semialg} that the length of a $2\otimes 2$ or $2\otimes 3$
separable state $\varrho$
is given by the maximum of ${\text{\rm rank}}\varrho$ and ${\text{\rm rank}}\varrho^\Gamma$,
and the length of an $m\otimes n$ separable state
may exceed the whole dimension $mn$ when $(m-2)(n-2)>1$.
Later, $3\otimes 3$ and $2\otimes 4$ separable states with length ten
have been constructed in \cite{{ha-kye-sep_2x4},{ha-kye-sep-face}}.
In the case of $\varrho_{a,b,c}$, we see that both the rank and the length are seven.
If a three qubit separable {\sf X}-state has rank four then it was shown in \cite{han_kye_phase} that
length is also four.
In case of a separable {\sf X}-state of  rank five (respectively six), the length is given by
$5$ or $6$ (respectively $6$, $7$ or $8$) \cite{han_kye_phase}.
It is not known that if there exists a three qubit state whose length exceeds the whole dimension $8$.
In this regard, three qubit {\sf X}-states might be the first possible target to find such states,
because we already know the separability condition.

\end{document}